\documentclass[11pt]{article}
\usepackage{ifpdf,latexsym,graphicx,url}
\usepackage[utf8]{inputenc}
\usepackage{subcaption}
\usepackage[margin=1in]{geometry}
\usepackage{amsthm, mathtools, amsfonts, grffile, wasysym, xspace, amsmath, float, ifthen, xifthen,thmtools, thm-restate}
\usepackage[table,xcdraw]{xcolor}

\newtheoremstyle{break}
  {\topsep}{\topsep}%
  {\itshape}{}%
  {\bfseries}{}%
  {\newline}{}%
\newtheoremstyle{breaknoitalics}
  {\topsep}{\topsep}%
  {\upshape}{}%
  {\bfseries}{}%
  {\newline}{}%

\newtheorem{claim}{Claim}[section]
\newtheorem{theorem}{Theorem}[section]
\newtheorem{lemma}{Lemma}[section]
\newtheorem{definition}{Definition}
\newtheorem{problem}{Problem}

\usepackage{tikz,tikz-cd}
\usetikzlibrary{shapes.geometric}

\usepackage{algorithm}
\usepackage{algpseudocode} 

{\makeatletter \gdef\fps@figure{!htbp}}

\let\realbfseries=\bfseries
\def\bfseries{\realbfseries\boldmath}

\newif\ifabstract
\abstracttrue
\newif\iffull
\ifabstract \fullfalse \else \fulltrue \fi

{\makeatletter
 \gdef\xxxmark{%
   \expandafter\ifx\csname @mpargs\endcsname\relax 
     \expandafter\ifx\csname @captype\endcsname\relax 
       \marginpar{xxx}
     \else
       xxx 
     \fi
   \else
     xxx 
   \fi}
 \gdef\xxx{\@ifnextchar[\xxx@lab\xxx@nolab}
 \long\gdef\xxx@lab[#1]#2{\textbf{[\xxxmark #2 ---{\sc #1}]}}
 \long\gdef\xxx@nolab#1{\textbf{[\xxxmark #1]}}
}

\usepackage[colorinlistoftodos]{todonotes}

\let\epsilon=\varepsilon
\let\subset=\subseteq
\def\degree{\operatorname{deg}}

\newcommand{\x}{\ensuremath{x}} 
\newcommand{\bignum}{\ensuremath{B}} 

\newcommand{\sufficientalg}{\mathcal{A}}

\DeclareMathOperator*{\maxcoeff}{maxcoeff}

\newcommand{\Ev}{\ensuremath{E}\xspace}

\newcommand{\ev}{\ensuremath{\mathrm{ev}}}

\newcommand{\ops}{\ensuremath{\mathsf{ops}}}
\newcommand{\variant}{\ensuremath{\mathsf{var}}}
\newcommand{\nats}{\ensuremath{\mathbb{N}}}

\newcommand{\AEC}[3]{\ensuremath{(#1,#2)}\textsc{-AEC-#3}}

\newcommand{\AECEO}[2]{\AEC{#1}{#2}{EO}}
\newcommand{\AECEL}[2]{\AEC{#1}{#2}{EL}}

\newcommand{\SG}[3]{\ensuremath{(#1,#2)}\textsc{-AEC-#3}}
\newcommand{\SGSTD}[2]{\SG{#1}{#2}{Std}}
\newcommand{\SGEO}[2]{\SG{#1}{#2}{EO}}
\newcommand{\SGEL}[2]{\SG{#1}{#2}{EL}}

\newcommand{\Partition}{\textsc{Partition}\xspace}
\newcommand{\PartitionEqual}{\textrm{\Partition-$n/2$}\xspace}
\newcommand{\ProductPartition}{\textsc{ProductPartition}\xspace}
\newcommand{\SquareProductPartition}{\textsc{SquareProductPartition}\xspace}
\newcommand{\ProductPartitionEqual}{\textrm{\ProductPartition-{$n/2$}}\xspace}
\newcommand{\SquareProductPartitionEqual}{\textrm{\textsc{SquareProductPartition}-{$n/2$}}\xspace}

\newcommand{\SetProductPartitionBound}[1]{\textsc{SetProductPartitionBound-\ensuremath{#1}\xspace}}

\newcommand{\TPartition}{\textsc{3-Partition}\xspace}
\newcommand{\TPartitionEqual}{\textsc{3-Partition}\-3\xspace}

\newcommand{\plustimes}{\mathpalette\plustimesinner\relax}
\newcommand{\plustimesinner}[2]{%
  \mathbin{\vphantom{+}\ooalign{$#1+$\cr\hidewidth$#1\times$\hidewidth\cr}}%
}

\newcommand{\defineproblemrefcomment}[5]{%
  \begin{problem}[\unboldmath #1]\rm~

  \unskip
  \begin{quote}
  \noindent\textbf{Instance:} #2

  \noindent\textbf{Question:} #3
  
  \ifthenelse{\isempty{#4}}{}{\noindent\textbf{Reference:} #4}
  
  \ifthenelse{\isempty{#5}}{}{\noindent\textbf{Comment:} #5}
  \end{quote}\end{problem}%
}

\newcommand{\defineproblemcomment}[4]{\defineproblemrefcomment{#1}{#2}{#3}{}{#4}}
\newcommand{\defineproblem}[3]{\defineproblemrefcomment{#1}{#2}{#3}{}{}}

\newcounter{section-preserve}
\newcounter{theorem-preserve}
\newcommand{\blank}[1]{}
\newtoks\magicAppendix
\magicAppendix={}
\newtoks\magictoks
\newif\iflater
\laterfalse
\long\def\later#1{\magictoks={#1}%
  \edef\magictodo{\noexpand\magicAppendix={\the\magicAppendix
    \the\magictoks%
  }}
  \magictodo}
\long\def\both#1{\magictoks={#1}%
  \edef\magictodo{\noexpand\magicAppendix={\the\magicAppendix
    \noexpand\setcounter{theorem-preserve}{\noexpand\arabic{theorem}}%
    \noexpand\setcounter{theorem}{\arabic{theorem}}%
    \noexpand\setcounter{section-preserve}{\noexpand\arabic{section}}%
    \noexpand\setcounter{section}{\arabic{section}}%
    \noexpand\let\noexpand\oldsection=\noexpand\thesection
    \noexpand\def\noexpand\thesection{\thesection}
    \noexpand\let\noexpand\oldlabel=\noexpand\label
    \noexpand\let\noexpand\label=\noexpand\blank
    \the\magictoks%
    \noexpand\setcounter{theorem}{\noexpand\arabic{theorem-preserve}}%
    \noexpand\setcounter{section}{\noexpand\arabic{section-preserve}}%
    \noexpand\let\noexpand\thesection=\noexpand\oldsection
    \noexpand\let\noexpand\label=\noexpand\oldlabel
  }}
  \magictodo
  \the\magictoks}
\def\magicappendix{\latertrue \the\magicAppendix}

\usepackage[unicode]{hyperref}

\newcommand{\CSAIL}{MIT CSAIL, Cambridge, MA, USA}
\newcommand{\MIT}{MIT, Cambridge, MA, USA}
\newcommand{\HVARD}{Harvard University, Cambridge, MA, USA}
\newcommand{\BU}{Boston University, Boston, MA, USA}

\title{Arithmetic Expression Construction}

\author{
Leo Alcock\thanks{\HVARD} \hspace{10pt}
Sualeh Asif\thanks{\MIT} \hspace{10pt}
Jeffrey Bosboom\thanks{\CSAIL} \hspace{10pt}
Josh Brunner\footnotemark[3] \hspace{10pt}
Charlotte Chen\footnotemark[2] \\
Erik D. Demaine\footnotemark[3] \hspace{10pt}
Rogers Epstein\footnotemark[3] \hspace{10pt}
Adam Hesterberg\footnotemark[1] \hspace{10pt}
Lior Hirschfeld\footnotemark[2] \\
William Hu\footnotemark[2] \hspace{10pt}
Jayson Lynch\footnotemark[3] \hspace{10pt}
Sarah Scheffler\thanks{\BU} \hspace{10pt}
Lillian Zhang\footnotemark[2]
}

\date{}

\begin{document}

\maketitle

\begin{abstract}
When can $n$ given numbers be combined using arithmetic operators from a given subset of $\{+,-,\times,\div\}$ to obtain a given target number?
We study three variations of this problem of \emph{Arithmetic Expression Construction}:
when the expression
(1)~is unconstrained;
(2)~has a specified pattern of parentheses and operators (and only the numbers need to be assigned to blanks); or
(3)~must match a specified ordering of the numbers (but the operators and parenthesization are free).
For each of these variants, and many of the subsets of $\{+,-,\times,\div\}$, we prove the problem NP-complete,
sometimes in the weak sense and sometimes in the strong sense.
Most of these proofs make use of a \emph{rational function framework} which proves equivalence of these problems for values in rational functions with values in positive integers.
\end{abstract}

\section{Introduction}

\emph{Algebraic complexity theory} \cite{AroraBarak,Gathen-survey}
is broadly interested in the smallest or fastest arithmetic circuit to compute
a desired (multivariate) polynomial.
An \emph{arithmetic circuit} is a directed acyclic
graph where each source node represents an input and every other node is an
arithmetic operation, typically among $\{+,-,\times,\div\}$, applied to
the values of its incoming edges, and one sink vertex represents the output.
One of the earliest papers on this topic is Scholz's 1937 study of
minimal addition chains \cite{scholzchains}, which is equivalent to finding
the smallest circuit with operation $+$ that outputs a target value~$t$.
Scholz was motivated by efficient algorithms for computing $x^n \bmod N$.
Minimal addition chains have been well-studied since; in particular,
the problem is NP-complete \cite{minadditionchains}.


Algebraic computation models serve as a more restrictive model of computation,
making it easier to prove lower bounds.  In cryptography, a common model is to
limit computations to a group or ring \cite{Maurer}.
For example, Shoup~\cite{Shoup} proves an exponential lower bound for discrete
logarithm in the generic group model, and Aggarwal and Maurer
\cite{AggarwalMaurer} prove that RSA is equivalent to factoring in the
generic ring model.
Minimal addition chains is the same problem as minimal group exponentiation in
generic groups, and thus the problem has received a lot of attention in
algorithm design \cite{FastExponentiation}.


In our paper, we study a new, seemingly simpler type of problem,
where the goal is to design an \emph{expression} instead of a \emph{circuit},
i.e., a \emph{tree} instead of a \emph{directed acyclic graph}.
Specifically, the main Arithmetic Expression Construction (AEC) problem
is as follows:

\defineproblem{\SGSTD{\mathbb{L}}{\ops} / \textrm{Standard}}
{A multiset of values $A = \{a_1, a_2, \dots, a_n\}\subseteq \mathbb{L}$ and a target value $t\in\mathbb{L}$.}
{Does there exist a parenthesized expression using any of the operations in $\ops$ that contains each element of $A$ exactly once and evaluates to $t$?}

The problem \SGSTD{\nats}{\{+,-,\times,\div\}} naturally generalizes
two games played by humans.
The 24 Game \cite{24game-wiki} is a card game dating back to the 1960s,
where players race to construct an arithmetic expression using four cards
with values 1--9 (a standard deck without face cards) that evaluates to 24.
In the tabletop role-playing game Pathfinder, the Sacred Geometry feat
requires constructing an arithmetic expression using dice rolls that evaluate
to one of a specified set of prime constants.

In this paper, we prove that this problem is NP-hard
when the input values are in $\mathbb N$
or the algebraic extension $\mathbb{N}[x_1, \ldots, x_k]$.%
\footnote{To clarify the notation: all values are in the field extension
$\mathbb{Q}(x_1, \ldots, x_k)$, but the
\emph{input} values are restricted to $\mathbb{N}[x_1, \ldots, x_k]$, i.e., have nonnegative integer coefficients.}


\subsection{Problem Variants and Results}\label{sec:variants-basic-orientation}

Expressions can be represented as trees with all operands at leaf nodes and operators at internal nodes using Dijkstra's shunting yard algorithm \cite{dijkstra1961algol}.  Similarly, an expression tree can be converted into a parenthesized expression by concatenating the operands and operators as they are encountered with an inorder traversal, adding an opening parenthesis when descending the tree and a closing parenthesis when ascending. 

\definecolor{op}{rgb}{.72,.64,.77}
\definecolor{leaf}{rgb}{.98,.91,.52}

\begin{figure}[H] 
  \centering
  \begin{tikzpicture}[
    every node/.style={minimum width=1.8em,inner sep=0,draw,circle},
    level distance=0.8cm,
    level 1/.style={sibling distance=3cm},
    level 2/.style={sibling distance=1.5cm},
    level 3/.style={sibling distance=1.5cm},
  ]
  \node[label=79,fill=op] at (0,0) {$\boldsymbol+$}
    child { node[label=77,fill=op] {$\boldsymbol\times$}
      child { node[fill=leaf] {$11$} }
      child { node[fill=leaf] {$7$} }
    }
    child { node[label=2,fill=op] {$\boldsymbol\div$}
      child { node[fill=leaf] {$4$} }
      child { node[label=2,fill=op] {$\boldsymbol-$}
        child { node[fill=leaf] {$3$} }
        child { node[fill=leaf] {$1$} }
      }
    };
  \end{tikzpicture}
  \caption[Arithmetic Expression Construction results]
    {An example expression tree for $7 \times 11 + (4 \div (3 - 1)) = 79$. The numbers above the internal nodes indicate their values.}
  \label{fig.exptree}
\end{figure}
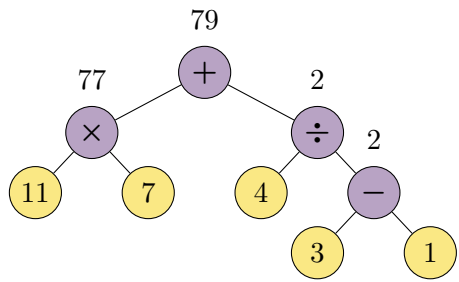

We also consider following two variants of AEC which impose additional constraints (specified by some data we denote by $D$) on the expression trees:

\defineproblem{\SGEL{\mathbb{L}}{\ops} / \textrm{Enforced Leaves}}
{A target value $t\in \mathbb{L}$ and a multiset of values $A = \{a_1,\ldots, a_n\}\subseteq \mathbb{L}$ with the leaf order encoded by $D:A\rightarrow [n]$.}
{Can an expression tree be formed such that each internal node has an operation from $\ops$, and the leaves of the tree are the list $A$ in order $D$, where the tree evaluates to $t$?}

\defineproblem{\SGEO{\mathbb{L}}{\ops} / \textrm{Enforced Operations}}
{A multiset of values $A = \{a_1, a_2, \dots, a_n\}\subseteq \mathbb{L}$, a target $t\in \mathbb{L}$, and an expression tree $D$ with internal nodes each containing an operation from $\ops$ and empty leaf nodes.}{Can the expression tree be completed by assigning each value in $A$ to exactly one leaf node such that the tree evaluates to $t$?}

The first variant fixes the ordering of leaf nodes of the tree, and asks whether an expression can be formed which reaches the target.  The second variant constrains the shape of the tree and internal node operations, and asks whether an ordering of the leaves can be found which evaluates to the target.
We represent all instances of these variants by triples $(A, t, D)$
where $A = \{a_1, a_2, \ldots a_n\}$ is a multiset of values,
$t$ is the target value, and $D$ is additional data for the instance:
a leaf ordering for EL, and an expression tree for EO.

In this paper, we prove hardness results in all of these variants by 
reduction from \Partition and related problems listed in Appendix \ref{sec:related}, and
develop polynomial or pseudopolynomial algorithms where appropriate.
Table~\ref{table:summary} summarizes our results.
In particular, we prove NP-hardness with $\mathbb{L} = \mathbb{N}$
for the Standard and EO variants for all subsets of operations $\{+,-,\times,\div\}$.
Note that all of these problems are in NP:
simply evaluate the expression given as a certificate.

\begin{table}[htbp]
\centering
\footnotesize
\tabcolsep=2pt
\begin{minipage}{\textwidth}%
\centering
\begin{tabular}{|c|c|c|c|} \hline
\textbf{Operations}  &  \textbf{Standard}  &  \textbf{Enforced Operations}  &  \textbf{Enforced Leaves} \\ \hline
$\{+\}$  &  $\in \text{P}$ (\S\ref{standard:plus})  &  $\in \text{P}$ (\S\ref{ops:plus})  &  $\in \text{P}$ (\S\ref{leaves:plus}) \\ \hline
$\{-\}$  &  \textbf{weakly NP-complete} (\S\ref{standard:minus})  &  \textbf{weakly NP-complete} (\S\ref{ops:minus})  &  \textbf{weakly NP-complete} (\S\ref{leaves:minus}) \\ \hline
$\{\times\}$  &  $\in \text{P}$ (\S\ref{standard:times})  &  $\in \text{P}$ (\S\ref{ops:times})  &  $\in \text{P}$ (\S\ref{leaves:times}) \\ \hline
$\{\div\}$  &  \textbf{strongly NP-complete} (\S\ref{standard:div})  &  \textbf{strongly NP-complete} (\S\ref{ops:div})  &  \textbf{strongly NP-complete}  (\S\ref{leaves:div}) \\ \hline
$\{+,-\}$  &  \textbf{weakly NP-complete} (\S\ref{standard:plus-minus})  &  \textbf{weakly NP-complete} (\S\ref{ops:plus-minus})  &  \textbf{weakly NP-complete} (\S\ref{leaves:plus-minus}) \\ \hline
$\{+,\times\}$  &  {weakly NP-complete (\S\ref{standard:plus-times})}  &  weakly NP-complete\footnote{Strong in all variables except the target $t$} (\S\ref{ops:plus-times})  &  {  \textbf{  weakly NP-complete}} (\S\ref{appendix-leaves:plus-times}) \\ \hline
$\{+,\div\}$  &  {weakly NP-complete (\S\ref{standard:plus-div})}  &  \textbf{strongly NP-complete} (\S\ref{ops:plus-div})  & Open \\ \hline
$\{-,\times\}$  &  weakly NP-complete (\S\ref{standard:minus-times})  &  \textbf{strongly NP-complete} (\S\ref{ops:minus-times})  & weakly NP-complete (\S\ref{appendix-leaves:minus-times})  \\ \hline
$\{-,\div\}$  &  {weakly NP-complete} (\S\ref{standard:minus-div})  &  \textbf{strongly NP-complete} (\S\ref{ops:minus-div})  &  Open \\ \hline
$\{\times,\div\}$  &  \textbf{strongly NP-complete} (\S\ref{standard:times-div})  &  \textbf{strongly NP-complete} (\S\ref{ops:times-div})  &  \textbf{strongly NP-complete} (\S\ref{leaves:times-div}) \\ \hline
$\{+,-,\times\}$  &  {  weakly NP-complete (\S\ref{standard:plus-minus-times})}  &  \textbf{strongly NP-complete} (\S\ref{ops:plus-minus-times})  &  { weakly NP-complete (\S\ref{leaves:plus-minus-times})} \\ \hline
$\{+,-,\div\}$  &  {weakly NP-complete} (\S\ref{standard:plus-minus-div})  &  \textbf{strongly NP-complete} (\S\ref{ops:plus-minus-div})  &  Open \\ \hline
$\{+,\times,\div\}$  &  {weakly NP-complete (\S\ref{standard:plus-times-div})}    &  \textbf{strongly NP-complete} (\S\ref{ops:plus-times-div})  &  {  weakly NP-complete} (\S\ref{leaves:plus-times-div}) \\ \hline
$\{-,\times,\div\}$  &  {    weakly NP-complete (\S\ref{standard:minus-times-div})}  &  \textbf{strongly NP-complete} (\S\ref{ops:minus-times-div})  &  Open \\ \hline
$\{+,-,\times,\div\}$  &  {weakly NP-complete (\S\ref{standard:plus-minus-times-div})}  &  \textbf{strongly NP-complete} (\S\ref{ops:plus-minus-times-div})  &  Open \\ \hline
\end{tabular}%




\caption{Our results for Arithmetic Expression Construction.  Bold font indicates NP-completeness results that are tight; for weakly NP-complete results, this means that we have a corresponding pseudopolynomial-time algorithm.  The proof is given in the section in parentheses. }

\label{table:summary}

\end{minipage}

\end{table}

Our first step is to show that, for any $k$ and $k'$, there is a polynomial-time
reduction from the $k$-variable variant to the $k'$-variable variant.
Such a reduction is trivial for $k \le k'$ by leaving the instance unchanged.
For the converse, we present the \emph{Rational Function Framework}
in Section~\ref{sec:rational}, which provides a polynomial-time construction of
a positive integer $B$ on an instance $I$
(i.e., set of values $\{a_i\}, t\in \mathbb{N}[x_1,\ldots, x_{k}]$)
such that replacing $x_{k} = B$ yields a solvable instance if and only if
$I$ is solvable. That is, for all variants $\variant \in \{\textsc{Std}, \textsc{EO}, \textsc{EL}\}$,
we obtain a simple reduction
$$\SG{\mathbb{N}[x_1, \ldots, x_k]}{\ops}{\variant} \rightarrow\SG{\mathbb{N}[x_1, \ldots, x_{k-1}]}{\ops}{\variant}$$
%
Because this reduction preserves algebraic properties, it yields
interesting positive results in addition to hardness results.
For example, this result demonstrates that
\SGSTD{\mathbb{N}[x_1,\ldots, x_k]}{\{+,-\}} has a pseudopolynomial-time
algorithm via a chain of reductions to \SGSTD{\mathbb{N}}{\{+,-\}}
which is equivalent to the classic \Partition problem.

\subsection{Notation}




Beyond the $I = (A, t, D)$ instance notation introduced above,
we often use the variable $E$ to denote an expression; the Standard variant is to decide whether $\exists E : E(A) = t$.
We also use ``$\ev(\cdot)$'' to denote the value of an expression at a node of an expression tree (i.e., the evaluation of the subtree whose root is that node).

\subsection{Outline of Paper}

In Section~\ref{sec:rational}, we describe the Rational Function Framework which demonstrates equivalence between AEC variants over different numbers of free variables. In Section~\ref{standard}, we present the structure theorem which will be used to prove hardness of the nontrivial cases of Standard and we present a proof of the full case with it. In Section~\ref{sec:OtherStandardCases}, we use two similar reductions to cover all remaining nontrivial cases of Standard. In Section~\ref{leaves}, we present the nontrivial proofs for the Enforced Leaves variant. In Section~\ref{ops}, we prove an interesting reduction for Enforced Operations. In Section~\ref{basic}, we present the remainder of our hardness proofs, which more straightforward and do not use the rational framework, along with pseudopolynomial algorithms for some weakly NP-hard problems. Appendix~\ref{sec:related} lists the problems we reduce from for our hardness proofs.

\section{Rational Function Framework}
\label{sec:rational}

\newcommand{\rats}{\ensuremath{\mathbb{Q}}\xspace}
\newcommand{\RatPolyxk}{\ensuremath{\rats(x_1, \ldots, x_k)}\xspace}

In this section, we present the rational function framework. This framework proves the polynomial-time equivalence of all Arithmetic Expression Construction variants with values 
as ratios of polynomials with integer coefficients, that is, \RatPolyxk, for differing $k$. This equivalence also allows us to restrict to $\nats[x_1,\ldots, x_k]$ and critically will make proving hardness for variants over $\nats$ easier by allowing us to reduce to $\nats[x_1,\ldots, x_k]$ versions.

\begin{theorem}
\label{theorem:full-ratl}
For all $\ops \subseteq \{+,-,\times,\div\}$, for all variants $\variant$, for all integers $k > 0$, there exists an efficient algorithm $\sufficientalg$ mapping instances $I$ to positive integers $\sufficientalg(I)$ such that a polynomial-time reduction 
$$\SG{\mathbb{Q}(x_1, \ldots, x_k)}{\ops}{\variant} \rightarrow\SG{\mathbb{Q}(x_1, \ldots, x_{k-1})}{\ops}{\variant}$$
is given by substituting $x_k = B$ in an instance $I$ for any $B \in \nats$ satisfying $B \ge  \sufficientalg (I)$. 
\end{theorem}



To formalize the idea of a ``big enough'' $B$, we introduce the concept of \emph{sufficiency} of integers for instances of AEC. Let $B$ be a positive integer and let $I$ be a $\SG{\mathbb{Q}(x_1,\ldots,x_k)}{\ops}{\variant}$ instance. Loosely, we consider $B$ to be \emph{$(I,\ops,\variant)$-sufficient} if substituting $x_k = B$ in instance $I$ creates a valid reduction on $I$, as in Theorem~\ref{theorem:full-ratl}.


We will shorten the terminology and call this \emph{$I$-sufficient} or \emph{sufficient for $I$} when $\ops$ and $\variant$ are clear from context. Theorem~\ref{theorem:full-ratl}~says there is an efficient algorithm that produces sufficient integers. Note that this definition is not yet rigorous. To remedy this we introduce the paired model of computation on rational functions.

In the paired model of computation, objects are given by pairs $(f,g)$ of integer-coefficient polynomials $f,g\in \mathbb{Z}[x_1,\ldots, x_k]$.  
Intuitively, the paired model simulates rational functions by $(f,g) \leftrightarrow f/g$.
We define operations ($+,-,\times,\div$) and equivalence relation ($\sim$) on pairs $(a,b)$ and $(f,g)$ as follows:
\begin{align*}
    (f,g) + (a,b) &= (fb+ga, gb)\\
    (f,g) - (a,b) &= (fb-ga, gb)\\
    (f,g)\times (a,b)&= (fa, gb)\\
    (f,g) \div (a,b) &= (fb, ga)\\
    (f,g)\sim (a,b) &\Leftrightarrow fb = ga
\end{align*}

As mentioned, the intuition is that $f$ is the numerator and $g$ is the denominator of a ratio of polynomials with integer coefficients.
The utility of the model is that it keeps track of rational functions as \textit{specific} quotients of integer coefficient polynomials. This will remove the ambiguity of representation of elements in $\mathbb{Q}(x_1,\ldots,x_n)$. Such a model allows us to make arguments about which polynomials can occur in the numerator and denominator of a rational function, such as by defining the range of these polynomials, as in the proof in Section~\ref{leaves:plus-times-div}. 

We can define Arithmetic Expression Construction in the paired model for all variants by changing target and values into pairs and using all the operations as defined above. 
An instance in the paired model is solvable if there exists a valid expression $E$ in values from $A$ and such that given $(f,g) = E(A)$, we have $(f,g)\sim (f_t, g_t) = t$.
We add restrictions on the solution in an additional variable, $D$.
For example, in enforced leaves, the entries of leaves of $E$ must be in the order specified by $D$,
and in enforced order, the expression $E$ is already specified and one must reorder $A$. The only difference is that we now compute in the paired model rather than with rational functions. 


Similarly, note that one can convert instances in the paired model to the nonpaired model via mapping entries $(f_i,g_i)\mapsto f_i/g_i$ and for a nonpaired model, one can always write $r\in \mathbb{Q}(x_1,\ldots, x_k)$ as $f_i/g_i$ where $f_i,g_i$ have integer coefficients.\footnote{Note that this representation is not unique!} A paired instance of AEC is solvable if and only if it's nonpaired variant is solvable. We now rigorously define sufficiency in Definition~\ref{definition:rigorous-I-sufficiency} and characterize its use in Lemma~\ref{lem:eval-replace-B}.




\begin{definition}
\label{definition:rigorous-I-sufficiency}
Let $B$ be a positive integer, and $I = (A, t=f_t/g_t, D)$ be an instance of \SG{\mathbb{Q}(x_1,\ldots,x_k)}{\ops}{\variant}. Represent $I$ in the paired model. 
Suppose that, for every evaluation $(f,g) = E(A)$ of a valid AEC expression $E$ (as restricted by $D$) in the paired model, the norms of the coefficients of $fg_t$ and $f_tg$ are all less than $B/2$. Then $B$ is $(I,\ops, \variant )$-sufficient.  
\end{definition}

\begin{restatable}{lemma}{ratlemeval}
\label{lem:eval-replace-B}
Given an instance $I = (A,t=f_t/g_t,D)$ of 
$\SG{\mathbb{Q}(x_1,\ldots, x_k)}{\ops}{\variant}$ and $\bignum\in \mathbb{N}$ which is $I$-sufficient as defined above. 
Let $\Ev (\cdot)$ be some expression from a valid $\ops$ expression tree according to $D$. Then, for every evaluation of $\Ev$ over the polynomials in $A$, one has:
\begin{multline*}
\Ev\left( \{(a_i(x_1,\ldots, x_k)\}_{a_i\in A} \right) = t(x_1,\ldots, x_k) \\ \Leftrightarrow \Ev \left(  \{a_i(x_1,\ldots, x_{k-1},B)\}_{a_i\in A}\right) = t(x_1,\ldots, x_{k-1},B).
\end{multline*}
\end{restatable}

\begin{proof}
Let $\Ev(A)$ denote $\Ev\left(\bigcup_{a_i \in A} a_i(x_1, \ldots, x_k)\right)$.  
For any evaluation $\Ev(A)$, we can write $\Ev(A) = f/g$ such that the norms of the coefficients of $fg_t$ and $f_tg$ are less than $B/2$. Note that $\Ev(A) = t$ if and only if $fg_t = f_tg$. 

Define $C$ as the set $\{ (\ell_1, \ldots, \ell_k) : x_1^{\ell_1} \cdots x_k^{\ell_k} \text{ has nonzero coefficient in one of }f_tg, fg_t \}$.
Then define the polynomial
$$s(x_1, \ldots, x_k) = \sum_{(\ell_1, \ldots, \ell_k) \in C}(B/2)x_1^{\ell_1} \cdots \x_k^{\ell_k}$$
Observe that all terms of $(fg_t+s)$ and $(f_tg+s)$ have positive coefficients of size between $0$ and $B-1$.
The Basis Representation Theorem  shows that if we have a polynomial $f(x) = \sum_i c_i x^i$ with all $c_i\in [0,B-1]$, then we can write $f(B)$ in base $B$, and recover all $c_i$.
A multivariate polynomial version of the same shows us that we can replace the last variable with $B$, and recover all the other coefficients in terms of the remaining variables.
Thus, we have 
\begin{multline*}
    (fg_t+s)(x_1,\ldots,x_{k}) = (f_tg+s)(x_1,\ldots, x_{k}) \\ \Leftrightarrow (fg_t+s)(x_1,\ldots, x_{k-1},B) = (f_tg+s)(x_1,\ldots, x_{k-1},B
\end{multline*}

Combining these we get the desired result:
\begin{align*}
    f/g = f_t/g_t &\Leftrightarrow fg_t = f_tg\\ 
    &\Leftrightarrow fg_t+s = f_tg +s\\ 
    &\Leftrightarrow (fg_t+s)(x_1,\ldots, x_{k-1},B) = (f_tg+s)(x_1,\ldots, x_{k-1},B) \\
    &\Leftrightarrow (fg_t)(x_1,\ldots, x_{k-1},B) = (f_tg)(x_1,\ldots, x_{k-1},B)
\end{align*}

This shows that $fg_t = f_tg$, which shows that $\Ev(A) = t$.
\end{proof}

 Essentially, this lemma shows that constructing $I$-sufficient integers efficiently is sufficient to prove our main theorem.  The rest of this section is dedicated to the polynomial-time construction of $I$ sufficient integers $B$ by an algorithm $\sufficientalg$.

Let
$$m(f)\coloneqq \binom{\deg(f) + k}{ \deg(f)}$$ 
where $m(f)$ is the maximum number of terms a $k$-variable polynomial $f$ of degree $\deg(f)$ can have. Let $\maxcoeff(f)$ denote the max of all of the $\emph{norms}$ of coefficients of $f$. That is, 
$$\maxcoeff(f) = \max_c \{|c|: c \text{ coefficient of } f\}.$$


Now we are ready to present an integer sufficient for an instance.

\begin{restatable}{lemma}{ratlemglobB}
\label{theorem: global-B}
Let $I = (A, t, D)$ be an instance of
$\SG{\mathbb{Q}(\x_1,\ldots,x_{k})}{\ops}{\variant}$.  
Then $$B = 2m(t) \maxcoeff(t) (2Mq)^n$$
is sufficient for $I$, where
$n = |A|$, 
$q \coloneqq \max_{f_i/g_i \in A} ( \maxcoeff(f_i), \maxcoeff(g_i))$ is the largest coefficient appearing in a paired polynomial within $A$, and
$M = \sum_{a_i \in A} m(a_i)$.
\end{restatable}

We prove this by first proving the following lemma bounding the coefficient blow up of the product of two polynomials and then inducting on this result to form our final $I$-sufficient $B$.

\textbf{Remark:} The algorithm presented in this proof gives a large $B$ that will give blowup sizes which are unnecessary for most AEC instances. One key use of sufficiency is to facilitate proofs with lower blowup. Often times we will have the following situation: We will give a reduction from a partition-type problem $P$ to $(\mathbb{Q}(x_i),\ops)$-AEC-$\variant$ and construct $(I,\ops,\variant)$-sufficient $B$ such that the composition
$$P\rightarrow \SG{\mathbb{Q}(x_1,\ldots,x_k)}{\ops}{\variant} \rightarrow \SG{\mathbb{N}}{\ops}{\variant}$$
is a valid reduction. 

\begin{lemma}
\label{lem:poly-product-bound}
Given two polynomials $a, b \in \mathbb{Z}[x_1,\ldots,x_{k}]$, let $h \in \mathbb{Z}[x_1, \ldots, x_k]$ be their product, $h = a * b$. 
The norm of each coefficient of $h$ is bounded by $$\min\{m(a), m(b)\} \cdot \maxcoeff(a) \cdot \maxcoeff(b)$$

\end{lemma}
\begin{proof}
Let 
\begin{align*}
    a(x_1,\ldots,x_{k}) &= \sum_{(\ell_l, \ldots, \ell_k)}c^{(a)}_{(\ell_1, \ldots, \ell_k)} x_1^{\ell_1}\cdots x_{k}^{\ell_k}, \ 
    b(x_1,\ldots,x_{k}) &= \sum_{(j_1,\ldots, j_k)} {c^{(b)}_{(j_1,\ldots, j_k)} x_1^{j_1}\cdots x_{k}^{j_k}}
\end{align*}
and let 
    $h(x_1, \ldots, x_k) =  \sum_{(i_1,\ldots, i_k)} {c^{(h)}_{(i_1,\ldots, i_k)} x_1^{i_1} \cdots x_{k}^{i_k}}$
be the product of $a$ and $b$. 

Then the coefficient for the $x_1^{i_1} \cdots x_k^{i_k}$ term of $h$ can be written as
\begin{align*}
	c^{(h)}_{(i_1,\ldots, i_k)} &= \sum_{(j_1,\ldots, j_k)} c^{(a)}_{(i_1-j_1,\ldots, i_k-j_k)}c^{(b)}_{(j_1,\ldots, j_k)} \\
	&\leq \sum_{(j_1,\ldots, j_k)} \maxcoeff(a) \maxcoeff(b) \\
	&\leq m(b) \maxcoeff(a) \maxcoeff(b)
\end{align*}
Symmetrically, we could have chosen to sum over the indices of $a$ rather than $b$, and so we also know that 
\begin{align*}
	c^{(h)}_{(i_1,\ldots, i_k)} &\leq m(a) \maxcoeff(a) \maxcoeff(b)
\end{align*}
Thus, we know that the norm of each coefficient of $h$ is bounded by 
$\min\{m(a), m(b)\} \cdot \maxcoeff(a) \cdot \maxcoeff(b)$
as desired.
\end{proof}

And now, the proof of lemma \ref{theorem: global-B}. 

\begin{proof}



Define $q$ to be the largest coefficient within all our paired functions:
$$q \coloneqq \max_{f_i/g_i \in A}\left(\maxcoeff(f_i), \maxcoeff(g_i)\right).$$ 
Recall that $m(a_i) \coloneqq \binom{\deg(a_i) + k }{ \deg(a_i)}$.
Let
$$M = \sum_{a_i \in A} m(a_i).$$


Let
$$B =  2 m(t) \maxcoeff(t)(2Mq)^n$$ 
where $n = |A|$.

First, observe that computing $B$ is efficient; using repeated squaring and schoolbook multiplication, $M = O(n)$ can be put to the $n$th power in $O((n \log n)^2)$ time.

We proceed to show that $B = 2 m(t) \maxcoeff(t)(2Mq)^n$ is sufficient for $I$. using strong induction on the number of operations in a subtree. Let $T_r$ be a valid evaluation on instance $I$. For any subtree $T$, let $A_T$ be the set of leaves in $T$, and let $p = |A_T|-1$ be the number of operations in this subtree. We will show that coefficients of $f_T$ and $g_T$ are bounded in magnitude by $(2Mq)^{p}$, where $(f_T,g_T) = \Ev(T)$.

If $p = 0$, then $T$  has one element $a_i(\x_1,\ldots,x_{k})$, and our result is trivially true since $B$ is larger than the max coefficient in any function in $A$ by construction. 

Next, for any $n > 0$ consider the left and right subtrees of $T$, $L$ and $R$. Let $(f_L, g_L)$ and $(f_R, g_R)$ be the corresponding polynomial pairs in $\mathbb{Q}(x_1, \dots x_k)$ that $L$ and $R$ evaluate to.
There are four possible operations that can combine $L$ and $R$. We show the $+$ case $(f_T, g_T) = (f_L,g_L) + (f_R,g_R) = (f_L g_R + g_L f_R, g_L g_R)$, but the other cases ($-$, $\times$, and $\div$) follow a very similar method. 

Let $j$ be the number of operations in $L$.
The bound on the coefficients of $f_T$ is:
\begin{align*}
    \maxcoeff(f_T) &= \maxcoeff(f_L g_R + g_L f_R)\\
    \begin{split}
    &= \max(m(f_L),m(g_R)) \maxcoeff(f_L)  \maxcoeff(g_R)\\
    & \qquad + \max(m(g_L),m(f_R)) \maxcoeff(g_L)  \maxcoeff(f_R)
    \end{split}\\
	&\leq  \max(m(f_L),m(g_R),m(g_L),m(f_R))(2Mq)^{j} (2Mq)^{p-j-1} \\
	&\leq  2M (2Mq)^{p-1}
	\\
	&\leq (2Mq)^{p}
\end{align*}

We can similarly show that $\maxcoeff(g_T) = \maxcoeff(g_Lg_R)$ is bounded by $(2Mq)^p$. 
The multiplication result relies on Lemma \ref{lem:poly-product-bound}.
\end{proof}

\subsection{Possible Generalizations to the Rational Framework}\label{sec:rationalobstructions}

In this section, we informally explore the possibility of extending the rational framework to the problems more general than expression construction, such as circuits. The generalization to circuits naturally becomes an arithmetic version of the Minimum Circuit Size Problem.  

The original Minimum Circuit Size Problem (MCSP) \cite{Kabanets99circuitminimization} asks if given a truth table and an integer $k$, can you construct a boolean circuit of size at most $k$ that computes the truth table; this problem has many connections throughout complexity theory.
A new variant, ``Arithmetic MCSP'' would ask if given $n$ values in $\{a_1,\ldots, a_n\}\subset \mathbb{L},$ within $0< k < n$ operations from $\{+,-,\times,\div\}$ can you construct a target $t\in\mathbb{L}$?%
\footnote{Note that since we can reuse values here, picking $k$ to be less than $n$ is the same as picking $k$ to be bounded by a fixed polynomial $p(n)$ by a padding argument. That is, you can reduce from this problem where you specify $k <p(n)$ to $k< n$ by padding any given instance $A$ with $\approx p(n)$ copies of $a_1$. This is similar to the proof that linear space simulation is PSPACE complete.} 
For $\mathbb{L} = \mathbb{Q}(x_1,\ldots, x_k)$, this problem asks whether a given rational function is constructable by an arithmetic circuit of size at most $k$ starting from a set of rational functions. 
It would be very useful if the rational framework could be adapted for Arithmetic MCSP; this would demonstrate an equivalence between the problem of circuit construction of rational functions and of reaching a rational number given input rational numbers. 

Unfortunately, the reduction methods provided above do not work naively for circuits: Given a polynomial-sized ``sufficient'' $B$ as presented, and a polynomial of the form $cx$, the term $(c^{2^k}x^{2^k})$ is formable by repeated squaring.   That is, we can form superpolynomial coefficents that will be bigger than $B$.  This removes the concept of ``sufficiency'' which is a key requirement for the rational framework as it is. 

On the bright side, the rational framework should work for Arithmetic Minimum Formula Constructions. Arithmetic formulae are expression trees with internal nodes operations $\{+,-,\times,\div\}$ except that one may use the input values in $A$ a flexible number of times. This is analogous to Boolean formulae; indeed, Minimum Boolean Formula problems \cite{BuchfuhrerBooleanFormulaMinimization,HemaspaandraBooleanFormulas} have also received significant attention. We can define Arithmetic Minimum Formula Construction as follows: Given multiset $A\subset \mathbb{L}$, target $t$, $0< k < n$, can you give a formula of size at most $k$ with values in $A$ which reaches a target $t\in \mathbb{L}$? 

The intuitive reason that the rational framework should hold in this case is because formulae still have a tree structure and the number of leaves is polynomial.  Thus, the same proofs in the rational framework will carry over. However, we expect the complexity and hardness proofs for this family of problems should be very different than those in this paper. All the reductions in this paper are from Partition-type problems, which allow for at most a single use of each input number. Hardness of this family of problems and generalizations of the rational function framework are interesting areas for further study.

\section{Main Standard AEC result from rational framework}
\label{standard}

In this section, we provide NP-hardness proofs for operations $\{+,\times\}\subset S\subset\{+,-,\times,\div\}$ of the Standard variant of Arithmetic Expression Construction. In Section~\ref{sec:OtherStandardCases} we give similar reductions that cover all other subsets of operations. 

All of these results use the rational function framework described in Section~\ref{sec:rational}. 

First, we outline some proof techniques that are used in this section to both combine proofs of results from differing sets of operations as well as simplify them. The first comes from the observation that if an instance of \SGSTD{\mathbb{L}}{S} is solvable, then for any operation set $S'\supset S$, the same instance will be solvable in \SGSTD{\mathbb{L}}{S'}. This allows us to bundle reductions to several AEC-\textsc{Std} cases simultaneously by giving a reduction ($R$) from some partition problem $P$ to \SGSTD{\mathbb{L}}{S} and proving that if any constructed instance is solvable in \SGSTD{\mathbb{L}}{S'}, the partition instance is also solvable.  That is, we have the following implications:
\begin{center}
\begin{tikzcd}
P\text{-instance } x \text{ Solvable } \ar[r, Rightarrow] & R(x) \text{ is } S \text{-Solvable}\ar[d,Rightarrow]\\
& R(x) \text{ is } S'\text{-Solvable}\ar[lu, Rightarrow]
\end{tikzcd}
\end{center}

\begin{theorem}\sloppy
Standard  {$\{+,\times\} \subset S \subset \{+, -, \times, \div\}$} is weakly NP-hard by reduction from \SquareProductPartitionEqual.
\end{theorem}
\label{standard:plus-minus-times-div}
\label{standard:plus-times}
\label{standard:plus-minus-times}
\label{standard:plus-times-div}

We spend the remainder of this section proving this theorem.

We will reduce from \SquareProductPartitionEqual (defined in Appendix~\ref{sec:related}) to \SGSTD{\mathbb{Z}[x,y,z]}{S}. On an instance $\{a_1,\ldots, a_n\}$ with all $a_i\ge 2$,\footnote{We can assume this property with loss of generality by replacing all $a_i$ with $2a_i$.} of \SquareProductPartitionEqual construct the following:

Let $$B_y = y - x^{n/2}\sqrt{\prod_i a_i};\  B_z = z - x^{n/2}\sqrt{\prod_i a_i}.$$ 

We then construct the instance of Arithmetic Expression Construction with input set $A = \left\{B_y, B_z \right\}\cup\{a_ix\}_i$ and target $t = yz$. Here the square root of the product of all $a_i$ is the value we want each partition to achieve, the polynomial $x^{n/2}$ will help us argue that we must multiply all of our $a_i$ values, and $B_y, B_z$ are gadgets which will force a partitioned tree structure as given by Theorem~\ref{structure-theorem}. Methods from Section~\ref{sec:rational} allow us to construct a reduction by replacing $x, y$, and $z$ with \emph{sufficient} integers $B_1$, $B_2$ and $B_3$.

If the \SquareProductPartitionEqual is solvable, then this AEC instance is solvable with operations $\{+,\times \} \subset S$. On the partition with equalized products, partition the $a_ix$ terms into corresponding sets and take their products to get two polynomials of value $x^{n/2}\sqrt{\prod_i a_i}$. Then form $(B_y+x^{n/2}\sqrt{\prod_i a_i})(B_z+x^{n/2}\sqrt {\prod_i a_i}) = yz$.

Next we prove the converse via contradiction by proving the following theorem that will be useful for the other AEC-\textsc{Std} cases. This theorem shows that any expression tree which evaluates to target $t \approx yz$ on an instance of similar structure to the constructed instance above must have a very particular partitioned structure described in Theorem~\ref{structure-theorem}. This will be the key to showing the soundness of our reduction. We use $\mathrm{ev}(T)$ to refer to the evaluation of the subtree rooted at node $T$.

Before stating Theorem~\ref{structure-theorem}, we first introduce the concept of $\mathbb{Q}(x)$-equivalence and give a couple of characterizations of it:
\begin{definition}\label{sim-f-definition}
Given a field $K$ with a subfield $F$, for $L_1,L_2\in K-F$, we say $L_1$ and $L_2$ are $F$-equivalent (written $L_1\sim_F L_2$) if by a sequence of operations between $L_1$ and elements of $F$ we can form $L_2$.
\end{definition}
The following lemma gives an alternate characterization of $\sim_F$: 
\begin{restatable}{lemma}{stdFequiv}
\label{F-equivalence-relation}
$\sim_F$ is an equivalence relation and $L_2\sim_F L_1$ if and only if for some $c_i,d_i\in F$ with $c_1d_2 -c_2d_1 \neq 0$,
$$L_2 = \frac{c_1L_1+d_1}{c_2L_1+d_2}$$
\end{restatable}

\begin{proof}\label{sim-f-proof}
Maps of the form $z\mapsto \frac{az+b}{cz+d}, ad-bc\neq 0$ are called linear fractional transformations. For a general reference on linear fractional transformations see \cite{linearfractional}. One can first note that for any operation with $c\in F-\{0\}$ with operations $\{+,-,\times,\div\}$ is a linear fractional transformation. E.g. 
$$c\div z = \frac{0z+c}{1z+0}, z+c = \frac{1z+c}{0z+1}$$
Another thing worth noting is that the composition of two linear fractional transformation is also a linear fractional transformation. In fact, composition of two fractional linear transformations described by matrices. 
$$\begin{pmatrix}a & b\\ c & d \end{pmatrix}, \begin{pmatrix}a' & b' \\ c' & d' \end{pmatrix}$$ is given by matrix composition.\footnote{In fact, the group of linear fractional transformation is $\mathrm{PSL}_2(F)$}
Thus, we can conclude that any $F$-equivalent element is of the form of linear fractional in $L_1$. However, to show that any linear fractional in $L_1$ is $F$-equivalent to $L_1$ we simply write an arbitrary linear fractional as a sequence of operations in $\{+,-,\times,\div\}$:
$$\frac{aL_1+b}{cL_1+d} = (a/d)L_1+b/d\text{ if }c=0, \frac{aL_1+b}{cL_1+d} = \frac{bc-ad}{c^2(L_1+d/c)}+a/c \text{ if }c\neq 0$$
\end{proof}

We will refer to $\mathbb{Q}(x)$ equivalence with respect to $\mathbb{Q}(x)$ as a subfield of $\mathbb{Q}(x,y,z)$.

We will also define the notion of degree of elements in the field $K(x)$, of rational functions with coefficients in $K$, and prove a lemma which will be useful in the proof of the structure theorem below. 

\begin{definition} Let $K$ be a field, consider the field $K(x)$ and the function $\deg_x: K(x)-0\rightarrow \mathbb{N}$ defined as $\deg_x(p/q) = \max (\deg_xp, \deg_xq)$ where $r = p/q$ is written in lowest terms. 
\end{definition}

\begin{lemma}\label{lemma:degree-props}
The function $\deg_x$ is subadditive in all the operations $+,-,\times, \div$. 
\end{lemma}

\begin{proof}
\begin{enumerate}
    \item \begin{align*}
        \deg_x\left (\frac{p_1}{q_1}*\frac{p_2}{q_2}\right )
        &= \deg_x\left (\frac{p_1p_2}{q_1q_2}\right ) \\
        &= \max\left (\deg_x\left (p_1p_2\right ),\deg_x\left (q_1q_2\right )\right ) \\
        &= \max\left (\deg_x\left (p_1\right )+\deg_x\left (p_2\right ),\deg_x\left (q_1\right )+\deg_x\left (q_2\right )\right ) \\
        &\le \max\left (\deg_x\left (p_1\right ),\deg_x\left (q_1\right )\right ) + \max\left (\deg_x\left (p_2\right ),\deg_x\left (q_2\right )\right )\\ 
        &= \deg_x\left (\frac{p_1}{q_1}\right ) + \deg_x\left (\frac{p_2}{q_2}\right ).
    \end{align*}
    \item \begin{align*}
        \deg_x\left (\frac{p_1}{q_1}\bigg/\frac{p_2}{q_2}\right ) 
        &= \deg_x\left (\frac{p_1}{q_1}*\frac{q_2}{p_2}\right ) \\
        &\le \deg_x\left (\frac{p_1}{q_1}\right )+\deg_x\left (\frac{q_2}{p_2}\right ) \\
        &= \deg_x\left (\frac{p_1}{q_1}\right )+\deg_x\left (\frac{p_2}{q_2}\right ).
    \end{align*}
    \item \begin{align*}
        \deg_x\left (\frac{p_1}{q_1}\pm \frac{p_2}{q_2}\right ) 
        &= \deg_x\left (\frac{p_1q_2 \pm p_2q_1}{q_1q_2}\right ) \\
        &\le \max\left (\deg_x\left (p_1q_2 \pm p_2q_1\right ),\deg_x\left (q_1q_2\right )\right ) \\
        &\le \max\left (\deg_x\left (p_1q_2\right ), \deg_x\left (p_2q_1\right ),\deg_x\left (q_1q_2\right )\right ) \\
        &\le \max\left (\deg_x\left (p_1\right ),\deg_x\left (q_1\right )\right ) + \max\left (\deg_x\left (p_2\right ),\deg_x\left (q_2\right )\right ) \\
        &= \deg_x \left  (\frac{p_1}{q_1}  \right ) + \deg_x  \left (\frac{p_2}{q_2} \right ).
    \end{align*}
\end{enumerate}
Also, if $\deg_x(p\pm q) \le 0$, $\deg_x(p*q) \le 0$, or $\deg_x(p/q) \le 0$, then $\deg_x(p) = \deg_x(q)$.
\end{proof}

We now state our structure theorem: 

\begin{theorem}
\label{structure-theorem}
For any $S\subset \{+,-,\times,\div\}$, let $I$ be a solvable  \SGSTD{\mathbb{Q}(x,y,z)}{S} instance with entries of the form $\{B_y, B_z\}\cup \{r_i(x)\}_i$ where $B_y\sim_{\mathbb{Q}(x)}y, B_z\sim_{\mathbb{Q}(x)} z$, and $r_i\in \mathbb{Q}[x]$ and target $t$ with $t \sim_{\mathbb{Q}(x)} yz$. Then any solution expression tree for $I$ has the form depicted in Figure~\ref{SGall}: The operation at the least common ancestor of leaves $B_y$ and $B_z$, denoted $N$, is $\times$ or $\div$, and $\mathrm{ev}(N) = (lyz)^{\pm 1}, l\in \mathbb{Q}(x)$. For $T_y, T_z$ the children of $N$ containing $B_y, B_z$ respectively, $\mathrm{ev}(T_y) = (ay)^{\pm 1},\mathrm{ev}(T_z) = (a'z)^{\pm 1}$, where $a,a'\in \mathbb{Q}(x)$.
\end{theorem}

\begin{figure}[H]
  \centering
  \begin{tikzpicture}[
    every node/.style={minimum width=1.8em,inner sep=0,draw,circle},
    label distance=-0.36cm,
    level distance=0.8cm,
    level 1/.style={sibling distance=6cm, level distance=1.3cm},
    level 2/.style={sibling distance=5cm, level distance=1.3cm},
    level 3/.style={sibling distance=3cm, level distance=0.8cm},
    level 4/.style={sibling distance=1.5cm, level distance=0.8cm},
  ]
  \node[fill=op] at (0,0) {}
    child { node[label=$(lyz)^{\pm 1}$,fill=op] {$N$}
      child { node[label=$(a y)^{\pm 1}$,fill=op] {}
        child { node[fill=op] {}
          child { node[fill=leaf] {$x a_1$} }
          child { node[fill=leaf] {$x a_2$} }
        }
        child { node[fill=leaf] {$B_y$} }
      }
      child { node[label=$(a' z)^{\pm 1}$,fill=op] {}
        child { node[fill=op] {}
          child { node[fill=leaf] {$B_z$} }
          child { node[fill=op] {}
            child { node[fill=leaf] {$x a_3$} }
            child { node[fill=leaf] {$x a_4$} }
          }
        }
        child { node[fill=leaf] {$x a_5$} }
      }
    }
    child { node[fill=op] {}
      child[level 4] { node[fill=leaf] {$x a_6$} }
      child[level 4] { node[fill=leaf] {$x a_7$} }
    };
  \draw (-8.25,-1.8) rectangle (-3.5,-4.7);
  \draw (-8.25,-1.8) node[below right,draw=none] { $T_y$ };
  \draw (-3.25,-1.8) rectangle (1.5,-5.5);
  \draw (-3.25,-1.8) node[below right,draw=none] { $T_z$ };
  \end{tikzpicture}
  \caption{Example expression tree for standard $\{+, -, \times, \div\}$.}
  \label{SGall}
\end{figure}
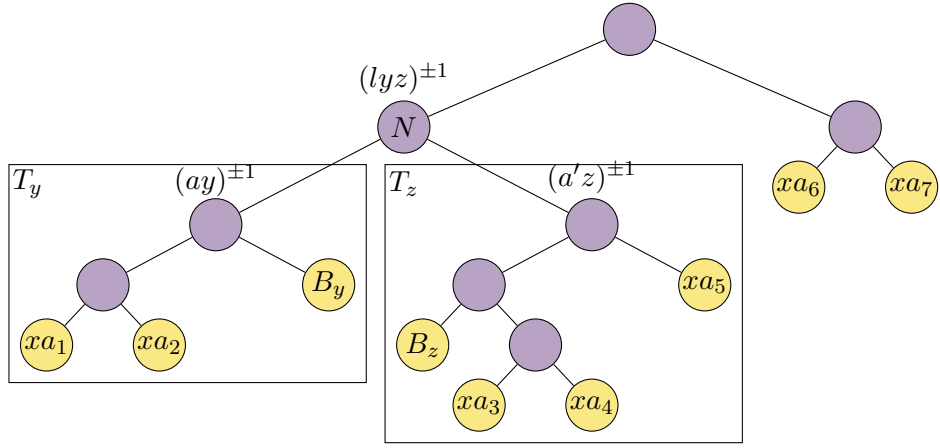

\begin{proof}
In our expression tree $T$, $N$ is the least common ancestor between $B_y$ and $B_z$. One has that 
$$ \mathrm{ev}({N}) = \frac{eyz + f}{gyz + h}, eh - gf \neq 0, e,f,g,h\in \mathbb{Q}(x)$$
since $\mathrm{ev}(N)$ is combined with a sequence of operations with elements in $\mathbb{Q}(x)$ to form $t$. That is, it is $\mathbb{Q}(x)$ equivalent to $yz$.

Let $T_y$ be the child of $N$ containing $B_y$ as a leaf and $T_z$ the child of $N$ containing $B_z$.
A priori we know 
$$\mathrm{ev}(T_y) = \frac{ay +b}{cy+d}, \mathrm{ev}(T_z) = \frac{a'z+b'}{c'z+d'}, \mathrm{ev}(N) = \frac{eyz + f}{gyz +h},$$
$$ad-bc \neq 0, a'd'-b'c' \neq 0, eh-fg \neq 0, a,b,c,d,a',b',c',d',e,f,g,h\in \mathbb{Q}(x)$$
by similar $\mathbb{Q}(x)$-equivalence arguments. The rest of the proof is casework done via trying out different operations at $N$. We will see that if the operation is $\times,\div$ then the evaluations must be of the form described in the statement of the theorem and that if the operation is $\pm$ then we reach a contradiction.

First we check the case that the operation at $N$ is $\times$. For this argument we'll reduce to a set of equations in $\mathbb{Q}(x)[y,z]$ and make some divisibility arguments using the fact that this is a unique factorization domain. 
\begin{align*}
    \frac{ay+b}{cy+d}\cdot \frac{a'z+b'}{c'z+d'} &= \frac{eyz+f}{gyz+h}\\
    \Rightarrow  (ay+b)(a'z+b')(gyz+h) &= (cy+d)(c'z+d')(eyz+f)
\end{align*}
If both $e,f\neq 0$, then $eyz + f$ is irreducible and since $eyz+f| (ay+b)(a'z+b')(gyz+h)$ we find that $eyz+f|gyz+h$ and $\frac{eyz+f}{gyz+h} = l\in \mathbb{Q}(x)$. However, this would contradict $\mathrm{ev}({N})\sim_{\mathbb{Q}(x)} yz$. We conclude that exactly one of $e,f$ is nonzero. A similar argument with $gyz+h$ allows us to conclude that at most one of $g,h$ is nonzero. We cannot have $g = 0$ and $e = 0$, or we would have $\mathrm{ev}(N)\in \mathbb{Q}(x)$. This reduces us to the case that $\mathrm{ev}(N) = (lyz)^{\pm 1}$. We now have one of the two cases:
\begin{align}
    (ay+b)(a'z+b') &= lyz(cy+d)(c'z+d')\\
    lyz(ay+b)(a'z+b')&= (cy+d)(c'z+d')
\end{align}
For the first case to hold one must have $c = c'=0$ for the degrees in $y$ and $z$ to match up. Given $c = c' = 0$, one must also have $b = b' = 0$ so that the right hand side of the equation is divisible by $yz$. A similar argument for the second case yields $a = a' = d = d'= 0$. For multiplication, this case is covered. If the operation is division, one gets the relation:
$$\frac{ay+b}{cy+d}\div \frac{a'y+b'}{c'y+d'} = \frac{ay+b}{cy+d}\cdot \frac{c'y+d'}{a'y+b'} = \frac{eyz+f}{gyz+h}$$
and the same argument follows through.

Next we show that the operation at $N$ can not be $+$:
\begin{equation}
\begin{split}
    \frac{ay+b}{cy+d} + \frac{a'z+b'}{c'z+d'} &= \frac{eyz+f}{gyz+h}\\
    ((ac' + a'c)yz + (ad'+b'c)y+(bc'+a'd)z+(bd'+b'd&))(gyz+h) \\ &= (cy+d)(c'z+d')(eyz+f)
\end{split}
\end{equation}
Starting with a similar divisibility argument, if $g,h\neq 0$, we find that $gyz+h$ is irreducible and that $gyz+h|eyz+f, \frac{eyz+h}{gyz+f}\in\mathbb{Q}(x)$. Thus either $g = 0$ or $h = 0$. 

Suppose $g = 0$.  Then we must have $e\neq 0$ to maintain $\mathrm{ev}(N)\sim_{\mathbb{Q}(x)} yz$. With nonzero $e$, one must have that $c = c' = 0$ so that the RHS of equation (9) has degree no bigger than the left hand side. The coefficient of $yz$ on the LHS of the equation is $(ac' + a'c)h = 0$ and the coefficient of $yz$ on the RHS is $edd'$ which must be nonzero and thus we get a contradiction.

Suppose $h = 0$. We must have $g,f\neq 0$ to maintain $\mathrm{ev}(N)\sim_{\mathbb{Q}(x)}yz$. The LHS of the equation is divisible by $yz$. Thus $yz|(cy+d)(c'z+d')(eyz+f)$ and this can only occur if $d = d' = 0$ and $c,c'\neq 0$. Expanding the equations now and looking at the coefficient of $yz$ in the LHS and RHS we find: $0\neq cc'f = g(bd'+b'd) = 0$. This concludes the proof of our helper theorem. 
\end{proof}

Now we will return to our proof of the soundness of the reduction to AEC-\textsc{Std}. Suppose that the constructed instance $I$ is solvable and the product partition instance is not solvable. Then for some $S \in \{\mathrm{leaves}(T_y) \cap \{a_i x\}, \mathrm{leaves}(T_z) \cap \{a_i x\}\}$, either
\begin{enumerate}
    \item $S$ contains $<n/2$ leaves $a_i x$.
    \item $S$ contains $n/2$ leaves $a_i x$ with product $\alpha x^{n/2}$ with $\alpha < \sqrt{\prod_ia_i}$. 
\end{enumerate}
WLOG let this set be $\mathrm{leaves}(T_y) \cap \{a_i x\}$. In the next two claims, we prove that in neither of these two cases can a subtree evaluate to an expression of the form $(ay)^{\pm 1}$ as Theorem~\ref{structure-theorem} requires.

\begin{claim}
\label{plus-times-std-case}
If $T_y$ contains $<n/2$ leaves $\{a_i x\}$ and $y' = y - x^{n/2}\sqrt{\prod a_i}$, then $\mathrm{ev}(T_y)$ is not of the form $(ay)^{\pm 1}$ for any $a \in \mathbb{Q}[x]$.
\end{claim}

\begin{proof}
The value of any subtree can be written in the form $\frac{p(x,y')}{q(x,y')}$ for polynomials $p$ and $q$. Let $\deg_x(\frac{p(x,y')}{q(x,y')}) = \max(\deg_x(p(x,y')),\deg_x(q(x,y')))$. This degree is subadditive for the four arithmetic operations ($+,-,\times, \div$), given by Lemma~\ref{lemma:degree-props}.  Also, if $\deg_x(p\pm q) \le 0$, $\deg_x(p*q) \le 0$, or $\deg_x(p/q) \le 0$, then $\deg_x(p) = \deg_x(q)$.

By induction, the degree in $x$ (respectively to $y'$) at a node $A$ is at most the number of leaves of $A$'s subtree of the form $a_ix$. This is true for the leaves ($\deg_x (a_ix) = 1$), and subadditivity proves it for the inductive step. 

Hence $\mathrm{ev}(T_y)$ has degree at most 1 in $y'$ and less than $n/2$ in $x$. If $\mathrm{ev}(T_y) = (ay)^{\pm 1} = (a(y'+x^{n/2})^{\pm1})$ for nonzero $a\in \mathbb{Q}(x)$, then it has degree at least $n/2$ in $x$, a contradiction.
\end{proof}

\begin{claim}
\label{plus-times-std-case-2}
If $T_y$ contains $n/2$ leaves $a_i x$ with $\prod_ia_i = \alpha <\sqrt{\prod_ia_i}$ and $y' = y - x^{n/2}\sqrt{\prod a_i}$, then $\mathrm{ev}(T_y)$ is not of the form $(ay)^{\pm 1}$ for any $a \in \mathbb{Q}(x)$.
\end{claim}

\begin{proof}
First, we rewrite our target $\mathrm{ev}(T_y)$ in terms of $y'$, yielding $\mathrm{ev}(T_y) = (a(y' + x^{n/2}\sqrt{\prod a_i}))^{\pm 1}$. We will first show that regardless of the value of $a$, the maximum coefficient of the rational function $\mathrm{ev}(T_y)$ is at least $\sqrt{\prod a_i}$. Note that since $y'$ is not in $\mathbb{Q}(x)$, $(y' + x^{n/2}\sqrt{\prod a_i})$ is an irreducible polynomial in $x$, so the denominator of $a$ will never cancel out with anything. Thus, we only consider the numerator of $a$. Consider the leading coefficient of the numerator of the product. This leading coefficient must be exactly the product of the leading coefficient of the numerator of $a$ and $x^{n/2}\sqrt{\prod a_i}$. Since the leading coefficient of the numerator of $a$ is an integer, it must be at least 1, so the leading coefficient of the numerator of $a(y' + x^{n/2}\sqrt{\prod a_i})$ must be at least $x^{n/2}\sqrt{\prod a_i}$.

From our reduction we have that all the $a_i$ are at least $2$, and the largest possible integer that can be generated from the $a_i$ and arithmetic operations is their product $\alpha$. Every coefficient of $\mathrm{ev}(T_y)$ is some combinations of arithmetic operations of the $a_i$ since it is comprised of the $a_ix$ and $y'$ and arithmetic operations. Thus, it is not possible for $\mathrm{ev}(T_y)$ to ever have a coefficient of at least $x^{n/2}\sqrt{\prod a_i}$. Thus, from the above argument it cannot be of the form  $(a(y' + x^{n/2}\sqrt{\prod a_i}))^{\pm 1}$.
\end{proof}

Note that the proof of this claim yields a reduction from $\SquareProductPartitionEqual$ to \SGSTD{\mathbb{Z}[x,y,z]}{S} for all $\{+,\times\}\subset S\subset \{+,-,\times,\div\}$. Using our rational function framework, we get a reduction from \SGSTD{\mathbb{Z}[x,y,z]}{S} to \SGSTD{\mathbb{Z}}{S} by replacements%
\footnote{Note that this denotes replacing with $B_i$ which are $I$-sufficient but since this is done via three reductions the instance $I$ changes. Therefore, when replacing with $B_2$, you need $B_2$ to be $I(B_1)$ sufficient (i.e., the instance $I$ with $x=B_1$ replaced). Similar requirements hold for $B_3$.} 
based on instance $I$ with 
$$x = B_1 = \sufficientalg(I), y = B_2 = \sufficientalg (I(B_1)), z = \sufficientalg(I(B_1,B_2)).$$

However, since the reduction is of the form 
$$\{y - \alpha x^{n/2}, z- \alpha x^{n/2}\}\cup \{a_ix\},$$ if we replace $B_2$ with $ B_2' =\max (B_2, 1+\alpha B_1^{n/2})$, and $B_3$ with $B_3' = \max (\sufficientalg(I(B_1,B_2')), 1+\alpha (\sufficientalg((B_1,B_2')))^{n/2})$ this will yield still sufficient $B_2, B_3$ such that the composition of these maps is a reduction from  $\ProductPartitionEqual$ to \SGSTD{\mathbb{N}}{S}.


\section{Other Standard AEC Results}
 \label{sec:OtherStandardCases}

To cover the rest of the cases for standard Arithmetic Expression Construction, we will give two more reductions of essentially the same structure and then use the Structure Theorem~\ref{structure-theorem} for reductions of this type and prove two analogues of the final claim in the proof from Claim~\ref{plus-times-std-case} above. 

\subsection{\texorpdfstring{$\{-,\times\}\subset S \subset \{+,-,\times,\div\}$}{\{−, ×\} ⊆ S ⊆ \{+, −, ×, ÷\}} AEC Standard} 

\label{standard:minus-times}
\label{standard:minus-times-div}

This case's reduction and hardness proof are identical to that of the proof above with only a couple of sign flips. On instances $\{a_i\}$ of \SquareProductPartitionEqual, produce instances $I$:
$$I = \left (\left\{y+x^{n/2}\sqrt{\prod_ia_i},z+x^{n/2}\sqrt{\prod_ia_i}\right\},\cup \left\{a_ix\right\}, t = yz \right)$$
of \SGSTD{\mathbb{N}[x,y,z]}{S}. If the product partition instance is solvable the constructed instance is solvable. Suppose the constructed instance is solvable in $\{+,-,\times,\div\}$ and the product partition instance is not solvable. 

First, we apply Theorem~\ref{structure-theorem} again, to obtain that the solution tree has the form of Figure~\ref{SGall}. The analogs of the Claims~\ref{plus-times-std-case} and~\ref{plus-times-std-case-2} with $y' = y+x^{n/2}\sqrt{\prod_ia_i}$ (instead of $y' = y-\sqrt{\prod_ia_i}$) and the proofs of the analogs are virtually identical. For an analog of Claim~\ref{plus-times-std-case}, the proof is identical since the proof only requires a degree in $x$ argument which is blind to the sign of $\prod_ia_i$. The proof of the analog of Claim~\ref{plus-times-std-case-2} refers to the maximal sizes of coefficients and thus goes through as well. This will reduce us down to the implication that if the constructed instance is solvable there must be a product partition and finishes the proof. 

\subsection{\texorpdfstring{$\{-,\div\},\{+,\div\},\{+,-,\div\}$}{\{−, ÷\}, \{+, ÷\}, \{+, −, ÷\}} AEC Standard}

\label{standard:minus-div}
\label{standard:plus-minus-div}
\label{standard:plus-div}

The form of the reduction in this section is similar to the previous reductions however we will reduce from \SquareProductPartition (in a way that simulates \SquareProductPartitionEqual) to \SGSTD{\mathbb{Q}(x,y,z)}{S} for $\{+,\div\}\subset S\subset \{+,-,\times,\div\}$. This will not affect the application of the Structure Theorem~\ref{structure-theorem} but our analogs of Claims~\ref{plus-times-std-case} and~\ref{plus-times-std-case-2} will be different. 

On an instance of \ProductPartition, $\{a_i\}_{i=1}^n$, construct the following instance: 
$$I = \left (\left\{y-x^{n+1}\sqrt{\prod_ia_i},z-x^{n+1}\sqrt{\prod_ia_i}\right\},\cup \left\{a_ix\right\}\cup\{x\}\ast [n+2], t = \frac{yz}{x^{2n}\prod_i a_i} \right).$$

We add $n+2$ $x$ monomials into the instance. If the product partition instance is solvable, you can partition $S = \{a_ix\}\cup \{x\}\ast [n+2]$ into two sets of $n+1$ monomials with product $x^{n+1}\sqrt{\prod_i a_i}$. Call these sets $S_1,S_2$. Divide $B_y$ by all but one of the elements of the form $x$ in $S_1$ and then add that element to it:
$$\frac{B_y}{a_1x\cdot a_2x\cdots } + x = \frac{y}{x^n\sqrt{\prod_ia_i}}$$
do the same with $B_z$ and then take the product of the output to reach the target. 
 
For the converse, assume the \SquareProductPartition instance is unsolvable but $I$ has a solution tree $T$: 

By the Theorem~\ref{structure-theorem}, we have a have two subtrees $T_y,T_z$ containing $B_y, B_z$ respectively which must evaluate to $(ay)^{\pm 1}, (a'z)^{\pm 1}$ respectively $a,a'\in\mathbb{Q}(x)$ with leaves disjoint subsets of $S = \{a_ix\}\cup \{x\}\ast [n+2]$. Suppose one of the two subtrees' leave's (WLOG $T_y$'s) has less than $n+1$ of the elements of $S$. Then this subtree cannot evaluate to $(ay)^{\pm 1}$ since has $\mathrm{deg}_x (ay)^{\pm 1}$ is higher than than the sum of $\mathrm{deg}_x$ of the subtree's leaves (where $\mathrm{deg}_x$ is defined relative to writing all rational functions with respect to $B_y, x$). Thus, this leaves the case where both $T_y$ and $T_z$ have $n+1$ monomials. 

\begin{claim}
Let $y' = y\pm \alpha x^{n+1}$,$\alpha\in \mathbb{N}$, and let $S$ be a (multi)set of $n+1$ elements of the form $a_ix, a_i\in \mathbb{N}$ satisfying 
the product of the elements of $S$ is $\beta x^{n+1}$ where $|\beta|<|\alpha|$.

Then all expressions in $S\cup y$ with operations $\{+,-,\div\}$ which yield $\deg_x = n+1$ are of the form:

\begin{enumerate}
    \item $\frac{y'}{\left(\prod_{a_i\in S}a_i\right) x^{n+1}}$\\
    \item $\frac{y'\pm \beta x^{n+1}}{\left(\prod_{i\neq k}a_i \right)x^n}$
\end{enumerate}

Namely, it will not be of the form $(ay)^{\pm 1}$ for $a\in \mathbb{Q}(x)$ since  $\deg_x (ay)^{\pm 1}\ge n+1$ (w.r.t. $(x,y')$). 
\end{claim}

\begin{proof}
We first note that any evaluation tree that evaluates to an element $v$ with $\deg_x v = n+1$ cannot have a nontrivial subtree with leaves consisting only of elements of the form $ax,a\in \mathbb{Z}$. This can be seen by noting that any operation between two elements $ax, bx$ is loses at least $1$ in $\deg_x$:
\begin{align*}
    ax\pm bx &= (a\pm b)x\\
    ax \div bx &= a/b
\end{align*}

This means that any evaluation tree evaluating to something of $deg_x = n+1$ must be a sequence of operations between $y'$ and the elements $\{a_ix\}$. Thus, our evaluation tree can be modeled as:
$$y'_0 = y', y_k' = {y}_{k-1}' \pm a_{k}x, \text{ or }{y}_k' = {y}_{k-1}'/(a_kx), \text{ or } {y}_k' = (a_kx)/{y}_{k-1}';\: \mathrm{ev}(T_y) = y_{n+1}'.$$

Suppose $y'_k = p(y',x)/(q(y',x)), k< n+1$,  with $\deg_x p > \deg_x q$, then the next operation loses at least one in $\deg_x$ and cannot yield an evaluation with $\deg_xv = n+1$:
\begin{align*}
    \frac{p}{q}\pm ax &= \frac{p\pm axq}{q};\\
    \frac{p}{q}\div ax &= \frac{p}{axq};\\
    ax\div \frac{p}{q}&= \frac{axq}{p}.
\end{align*}

Thus we have the invariant that $\deg_x q\ge \deg_x p$ for $y'_i, i <n+1$. However, from the above equations, we can also note that if $\deg_x q \ge \deg_x p$, only the operation $p/q \div ax$ maintains this property. Thus we conclude that 
$$y'_n = \frac{y'}{(\prod_{i\neq k}{a_i})x^n}$$
given a sequence of $n$ $a_ix$ divisions with $a_ix \in S$. The last operation can be an addition, subtraction or division by $a_ix$ and any of these operations will give the two cases from the statement of the claim.  
\end{proof}

Note that if the product partition instance is not solvable, then one of the subtrees has monomials $a_ix$ whose product is less than $\sqrt{\prod_ia_i}x^{n+1}$ and thus by this claim and the structure theorem we finish the proof of the converse. 

We note at this point that the exact same argument will apply for $S = \{-,\div\}$ where you flip the signs in $B_y,B_z$ and proceed similarly. 

The last piece of housekeeping is translating this reduction to AEC with entries in $\mathbb{N}$. As it stands, we can use the algorithm $\sufficientalg$ to produce $B_i$ which will create a valid reduction to AEC with entries in $\mathbb{Q}$. In the $\{-,\div\}$ construction these replacements will yield positive coefficients already. For $\{+,\div\}$, on an instance $I$ we can do replacements:
$$x = B_1 = \sufficientalg(I), y = B_2' = \max \left(B_2, 1 +B_1^{n+1}\sqrt{\prod_ia_i}\right), z = B_3' = \max \left(B_3, 1 + \sqrt{\prod_ia_i}B_1^{n+1}\right)$$ to achieve positivity. After these replacements, the only element that may still not be an integer is 
$$t = \frac{yz}{(\prod_i a_i) x^{2n}}$$
and one can do a final replacement $z = B_3'' = (\prod_i a_i) B_1^{2n}B_3'$. This finishes the hardness proofs for all of the standard AEC cases.

\section{Enforced Leaves AEC Results from Rational Framework}
\label{leaves}
Recall that an instance of the Enforced Leaves (EL) AEC variant has a fixed ordering of leaves (operands), and the goal is to arrange the internal nodes of the expression tree such that the target $t$ is the result of the tree's evaluation.
In this section, we present hardness proofs for operation sets $\{+,\times\}, \{+, -, \times\}, \{-, \times \}, \{+,\times, \div\}$ of the Enforced Leaves variant.


\subsection{Weak NP-completeness of AEC Enforced Leaves \texorpdfstring{$\{+,\times\}$}{\{+, ×\}}}
\label{appendix-leaves:plus-times}

\begin{claim}
$\AECEL{\nats[x]}{\{+,\times\}}$ is weakly NP-hard.
\end{claim}

Before we prove this claim, we state and prove some useful lemmas that utilize the following $\AECEL{\nats[x]}{\{+,\times\}}$ instance structure:

Given an instance of \Partition-$n/2$ with set of positive integers $A = \{a_1, a_2, \ldots, a_n\}$, let $I_A$ be an instance of $\AECEL{\nats[x]}{\{+,\times\}}$ with 
polynomials of the form $a_i x^3$ interspersed with $n-1$ polynomials $x$ in the leaf order 
$$a_1 x^3 \enspace x\enspace a_2 x^3 \enspace x\enspace a_3 x^3 \cdots x \enspace a_n x^3$$
and with target $t(x) = (x^4+x^3)\frac{\sum_ia_i}{2} + (\frac{n}{2} -1)x$.

\begin{figure}[!tbp]
  \centering
  \begin{minipage}[b]{0.42\textwidth}
    \centering
    \begin{tikzpicture}[
      every node/.style={minimum size=1.8em,inner sep=0,draw,ellipse},
      label distance=-0.1cm,
      level distance=1cm,
      level 1/.style={sibling distance=4cm},
      level 2/.style={sibling distance=1.5cm},
      level 3/.style={sibling distance=1.5cm},
    ]
    \node[label={$m$, $\degree(m) \geq 5$},fill=op] at (0,0) {$\times$}
      child { node[label={$c$, $\degree(c) \geq 1$},fill=leaf] {$x$} }
      child { node[label={$d$, $\degree(d) \geq 4$},fill=op] {$\times$}
        child { node[label={$g_1$},fill=leaf] {$a_i x^3$} }
        child { node[label={$g_2$},fill=leaf] {$x$} }
      };
    \end{tikzpicture}
    \caption{Proof of Lemma \ref{lem:leaves:plus-times1}: a $\times$ node cannot have a descendant $\times$ node.}
    \label{fig:leaves:plus-times1}
  \end{minipage}
  \hfill
  \begin{minipage}[b]{0.5\textwidth}
    \centering
    \begin{tikzpicture}[
      every node/.style={minimum size=1.8em,inner sep=0,draw,ellipse},
      label distance=-0.1cm,
      level distance=1cm,
      level 1/.style={sibling distance=4cm},
      level 2/.style={sibling distance=1.5cm},
      level 3/.style={sibling distance=1.5cm},
    ]
    \node[label={$m$, $\degree(m) \geq 6$},fill=op] at (0,0) {$\times$}
      child { node[label={$c_1$, $\degree(c_1) \geq 3$},fill=op] {$+$}
        child { node[fill=leaf] {$a_{i-1} x^3$} }
        child { node[fill=leaf] {$x$} }
      }
      child { node[label={$c_2$, $\degree(c_2) \geq 3$},fill=op] {$+$}
        child { node[fill=leaf] {$a_i x^3$} }
        child { node[fill=leaf] {$x$} }
      };
    \end{tikzpicture}
    \caption{Proof of Lemma \ref{lem:leaves:plus-times2}: a $\times$ node cannot have two children $+$ nodes.}
    \label{fig:leaves:plus-times2}
  \end{minipage}
\end{figure}

\begin{lemma}
\label{lem:leaves:plus-times1}
All solutions of $I_A$ that meet the target have the property that no $\times$ operator node has a $\times$ operator descendant node.
\label{eq:lem-leaves-PM-1}
\end{lemma}

\begin{proof}[Proof of Lemma~\ref{lem:leaves:plus-times1}]
Refer to Figure~\ref{fig:leaves:plus-times1}.
With the operations $\{+, \times\}$ and leaves that are positive powers of $x$ with positive coefficients, the value at a node cannot be a polynomial of higher degree than any of its ancestors.
For the sake of contradiction, consider some node $m$ with the $\times$ operator and a descendant $d$ also with the $\times$ operator. Let  $g_1, g_2$ be the children of $d$, and let $c$ be the child of $m$ that is not $d$ or an ancestor of $d$.
Then $c$ must evaluate to at least $x$; similarly, one of $g_1$ and $g_2$ evaluates to at least $x$ and the other to at least $x^3$. However this implies that the parent $\times$ node evaluates to at least $x^5$, which is greater than the target. Therefore, any node with the $\times$ operator can only have leaves or $+$ operators as descendants. 
\end{proof}

\begin{lemma}
\label{lem:leaves:plus-times2}
All solutions to $I_A$ that meet the target have the property that there is at least one leaf child of all internal $\times$ nodes.
\end{lemma}

\begin{proof}[Proof of Lemma~\ref{lem:leaves:plus-times2}]
Refer to Figure~\ref{fig:leaves:plus-times2}.
Assume for the sake of contradiction that $m$ is a $\times$ operator node and
neither of its children $c_1,c_2$ are leaf nodes.
Then, by Lemma~\ref{lem:leaves:plus-times1}, both children must be $+$ operator nodes and thus have at least two descendants;
Because of the alternating ordering of $x$ and $a_i x^3$ leaves, both of these subtrees must have an $a_i x^3$ term. The product of these sums must then have a term of order $x^6$, which is a contradiction. 
\end{proof}

\begin{figure}[h!tbp]
  \centering
  \begin{subfigure}[b]{0.3\textwidth}
    \centering
    \begin{tikzpicture}[
      every node/.style={minimum size=1.8em,inner sep=0,draw,ellipse},
      label distance=-0.1cm,
      level distance=1cm,
      level 1/.style={sibling distance=3cm},
      level 2/.style={sibling distance=1.5cm},
    ]
    \node[label={$a_i x^4$},fill=op] at (0,0) {$\times$}
      child { node[fill=leaf] {$x$} }
      child { node[fill=leaf] {$a_i x^3$} };
    \end{tikzpicture}
    \caption{Case 1: Both children are leaf nodes.}
  \end{subfigure}%
  \hfill
  \begin{subfigure}[b]{0.3\textwidth}
    \centering
    \begin{tikzpicture}[
      every node/.style={minimum size=1.8em,inner sep=0,draw,ellipse},
      label distance=-0.1cm,
      level distance=1cm,
      level 1/.style={sibling distance=3cm},
      level 2/.style={sibling distance=1.5cm},
    ]
    \node[label={$c_0 x^4 + c_1 x^2$},fill=op] at (0,0) {$\times$}
      child { node[fill=leaf] {$x$} }
      child { node[label=$c_0 x^3 + c_1 x$,fill=op] {$+$}
        child { node[fill=leaf] {$a_i x^3$} }
        child { node[fill=leaf] {$x$} }
      };
    \end{tikzpicture}
    \caption{Case 2: One child is a leaf node with value $x$.}
  \end{subfigure}%
  \hfill
  \begin{subfigure}[b]{0.3\textwidth}
    \centering
    \begin{tikzpicture}[
      every node/.style={minimum size=1.8em,inner sep=0,draw,ellipse},
      label distance=-0.1cm,
      level distance=1cm,
      level 1/.style={sibling distance=3cm},
      level 2/.style={sibling distance=1.5cm},
    ]
    \node[label={$c_2 x^6 + c_3 x^4$},fill=op] at (0,0) {$\times$}
      child { node[fill=leaf] {$a_i x^3$} }
      child { node[label=$c_0 x^3 + c_1 x$,fill=op] {$+$}
        child { node[fill=leaf] {$x$} }
        child { node[fill=leaf] {$a_{i+1} x^3$} }
      };
    \end{tikzpicture}
    \caption{Case 3: One child is a leaf node with value $a_ix^3$.}
  \end{subfigure}
  \caption{Proof of Lemma~\ref{lem:leaves:plus-times3}.}
  \label{fig:leaves:plus-times3}
\end{figure}

\begin{lemma}
\label{lem:leaves:plus-times3}
In all solutions to $I_A$ that meet the target, any internal $\times$ operator node has exactly two children, both of which are leaves.
\end{lemma}

\begin{proof}[Proof of Lemma~\ref{lem:leaves:plus-times3}]
Refer to Figure~\ref{fig:leaves:plus-times3}.
By Lemma~\ref{lem:leaves:plus-times2}, there are three possible child pairs under a $\times$ operator: (i) two leaves $x \times ax^3$ or $ax^3 \times x$, (ii) one $ax^3$ leaf and a $+$ operator, and (iii) one $x$ leaf and a $+$ operator. 
We will show that the first case is the only one that does not provide a contradiction. 

In the second case, there is at least one $ax^3$ leaf under the $+$ operator, so the degree of the polynomial evaluated at the $\times$ operator would be at least $6$, a contradiction.

Consider the third case. Let $m$ be the $\times$ operator and let $e, d$ be its children where $e$ is the $x$ leaf and $d$ is the $+$ operator. Since, by Lemma~\ref{lem:leaves:plus-times1}, $\times$ operator nodes cannot have $\times$ operator descendants, all descendants of $d$ must be leaves or $+$ operators. Due to the alternating leaf order, $d$ must have an $x$ descendant and an $a_i x^3$ descendant, so the evaluation of $d$ is of the form $c_0 x^3 + c_1 x$ where $c_0$ and $c_1$ are positive integers. Thus the evaluation of $m$ is of the form $c_0 x^4 + c_1 x^2$. We show that it is impossible to reach the target $t(x)$ if an internal node in the evaluation contains this expression.


By Lemma~\ref{lem:leaves:plus-times1}, no ancestor node $p$ of $m$ can be a $\times$ operator as $m$ is a $\times$ operator. Thus all ancestors of $m$ are $+$ operators. Thus the evaluation at $p$ must have an $x^2$ term since we are summing multiple polynomials with $x^2$ terms with non-negative coefficients. 
This remains true when $p$ is the root, so the evaluation of the entire expression tree must contain an $x^2$ term, which is a contradiction.
\end{proof}


\begin{proof}[Proof of Claim \ref{appendix-leaves:plus-times}]
This proof proceeds by reduction from \Partition- {$n/2$}.

Let $A$ be the set of positive integers for \Partition- {$n/2$}, and let $I_A$ consist of the ordering and target defined above.

If the instance of \Partition-$n/2$ has a solution, we know two complementary subsets $A_1, A_2$ of $A$ exist such that $\sum A_1 = \sum A_2 =  \frac{\sum A}{2}$ and $A_1$ contains $a_1$. Then in the AEC instance, for $a_i$ in $A_1$ (other than $a_1$) we assign operations so that $\cdots + (x + a_i x^3) + \cdots$, and for $a_i$ in $A_2$ we assign operations so that $\cdots + (x \times a_i x^3) + \cdots$. 

Let the symbol $\plustimes$ represent the choice of either $+$ or $\times$ depending on the set $a_i$ belongs to. Evaluating our expression, we get
\begin{align*}
    a_1 x^3 + (x \plustimes a_2 x^3) + \cdots + (x \plustimes a_n x^3) &= a_1x^3 + \sum_{a_i \in A_1; i \neq 1}(x+ a_i x^3) + \sum_{a_i \in A_2} a_i x^4 \\
    &= (x^4+x^3)\frac{\sum_ia_i}{2} + \left(\frac{n}{2} -1\right)x  \\
    &= t(x).
\end{align*}
Therefore, our constructed instance of AEC has a solution if the original instance of \Partition-$n/2$ has a solution.




To prove the other direction, consider a solution to $I_A$ that meets $t(x)$.
Recall that there are $(n-1)$ $x$ leaves of the instance.  The only way to achieve the $(n/2-1)$ $x$ terms in the target is for exactly $(n/2-1)$ of those leaves to have $+$ operator parents instead of $\times$, since otherwise the $x$ term is multiplied by an $ax^3$ node and (since we have only $+$ and $\times$ with non-negative coefficients) this cannot be an $x$ term in the evaluation.
Thus, the remaining $(n/2)$ $x$ terms must have $\times$ operator parents.
By  Lemma~\ref{lem:leaves:plus-times3}, the only $\times$ operators in the expression tree are parents of two leaves, and so are of the form $(x \times ax^3)$ or $(ax^3 \times x)$. 

If we let $A_2$ be the set of $a_i$ such that $a_ix^3$ was as child of a $\times$ operator, and $A_1$ be the set of all other $a_i$, then we have divided 
$A$ into two complementary subsets $A_1$ and $A_2$ such that $|A_1|=|A_2|=n/2$. Further, since the target was achieved, $x^4\sum A_2 +x^3\sum A_1 + (|A_1| -1 )x = (x^4+x^3)\frac{\sum_ia_i}{2} + (\frac{n}{2} -1)x $ so $\sum A_1 =\frac{\sum_ia_i}{2} = \sum A_2$. Thus $A_1$ and $A_2$ provide a valid partitioning for $\textsc{Partition-}n/2$, so the original instance of $\textsc{Partition}$-$n/2$ has a solution if our constructed instance of $\{+, \times \}$-\textsc{EL} has a solution.

As our instance of $\{+, \times\}$-\textsc{EL} has a solution if and only if the instance of $\textsc{Partition}$ has a solution, we have found a valid reduction from $\textsc{Partition}$ to $\{+, \times\}$-\textsc{EL}. As our reduction takes polynomial time and \Partition-$n/2$ is weakly NP-hard, $\{+, \times\}$-\textsc{EL} must also be weakly NP-hard.
\end{proof}

\begin{claim}
\label{appendix-leaves:pseudopoly-PT}
$\AECEL{\nats}{\{+,\times\}}$ is weakly NP-complete.
\end{claim}
\begin{proof}
We know by Claim \ref{appendix-leaves:plus-times} combined with Theorem \ref{theorem:full-ratl} that $\AECEL{\nats}{\{+,\times\}}$ is weakly NP-hard.  To show that this hardness is tight, we provide a pseudopolynomial algorithm for it.


Suppose that the enforced leaf ordering is $a_1, a_2, \ldots, a_n$ (all in $\nats$), and the target is $t \in \nats$. 
Let $F_{i,j}$ be the set of all possible values in $[t]$ that can be attained via operations for the values ordered $a_i, \ldots, a_j$, for $1 \le i \le j \le t$. The instance has a solution if and only if $F_{1,n}$ includes $t$.

We can use dynamic programming on the $F_{i,j}$ sets to compute $F_{1,n}$.
Initialize $F_{i,i} = \{a_i\}$ for all $i \in [n]$. Then each set $F_{i,j}$ with $j>i$ can be computed as the union
$$F_{i,j} = \bigcup_{k=i}^{j} 
\{\ell + r, \ell r ~\vert~ \ell \in F_{i,k}, r \in F_{k+1,j}\} \cap [t]. $$
%
Note that each of these sets is at most size $t$; we do not need to keep track of values larger than $t$ since they cannot be combined using $+$ and $\times$ operations to reach the target $t$.
Therefore, computing each set of new values from $F_{i,k}$ and $F_{k+1,j}$ takes time $O(t^2)$.  To get each $F_{i,j}$ we take the union over $O(n)$ such computed sets.  To ``reach'' $F_{1,n}$ in this way we must compute $O(n^2)$ values of $F_{i,j}$, so the total runtime of this protocol is $O(n^3 t^2)$.

The existence of this pseudopolynomial algorithm shows that $\AECEL{\nats}{\{+,\times\}}$ is not strongly NP-hard, and so combined with our hardness result, we have shown that the problem is weakly NP-complete, as desired.
\end{proof}

\subsection{Weak NP-hardness of AEC Enforced Leaves \texorpdfstring{$\{+,-, \times\}$}{\{+, −, ×\}}}
\label{leaves:plus-minus-times}
\label{appendix-leaves:plus-minus-times}

We present a proof for the weak NP-hardness of $\AECEL{\mathbb{N}[x,y]}{\{+, -, \times\}}$. Using the technique described in Section \ref{sec:rational}, this also proves NP-hardness of $\AECEL{\mathbb{N}}{\{+, -, \times\}}$.

Our proof is a reduction from \SetProductPartitionBound{K}. This strongly NP-hard problem asks if given a set (without repetition) of positive integers $A = \{a_1, a_2, \dots, a_n\}$ where all $a_i > K$ and all prime factors of all $a_i$ are also greater than $K$, we can partition $A$ into two subsets with equal products.  The problem is also defined formally in Appendix \ref{sec:related}.

\begin{restatable}{claim}{ELclaim}
\label{claim:EL-plus-minus-times}
$\AECEL{\mathbb{N}[x,y]}{\{+, -, \times\}}$ is weakly NP-hard.
\end{restatable}

This statement is proved via reduction from \SetProductPartitionBound{3}.  Let the instance be $A = \{a_1, \ldots, a_n\}$, where all prime factors of all $a_i \in A$ (and all $a_i$ themselves) are greater than 3. 

Let $L = 2 \prod_{i \in [n]} a_i$. 
Let $p_1, \ldots, p_n$ be unique primes greater than 3 that are coprime to $\prod_{i \in [n]} a_i (a_i^2 + 1)(a_i^2 - 1)$.
For each $a_i$, construct integer-coefficient $y$ terms $b_i = \frac{1}{2} L(a_i + \frac{1}{a_i})p_i y$ and $c_i = \frac{1}{2} L(a_i - \frac{1}{a_i})p_i y$.
Observe that $b_i + c_i = (La_i)p_iy$ and $b_i - c_i = (L/a_i)p_i y$.
Also note that they both have integer coefficients because $a_i | L$ for all $i$.

Now, construct instance $I_A$ of $\AECEL{\nats[x,y]}{\{+,-,\times\}}$ which has target polynomial $t(x,y) = L^n y^n x^{n-1} \prod_{i \in [n]} p_i$, and the following order of leaves:
$$
\label{eqn:leaves:plus-minus-times:ordering}
b_1  ~\enspace~ c_1  ~\enspace~ x ~\enspace~ b_2  ~\enspace~ c_2  ~\enspace~ x~ \enspace~ \cdots ~\enspace~ x~ \enspace~ b_n ~ \enspace~ c_n.
$$

If an instance of this product partition variant is solvable, then the constructed instance evaluates to $t(x,y) = L^n x^{n-1} y^n \prod_{i \in [n]} p_i$ when we have $(b_i + c_i)$ for $a_i$ in one partition and $(b_i - c_i)$ for $a_i$ in the other, and the $\times$ operator at every other node.
The partition corresponds to whether the $a_i$ was written as a difference or a sum.

We must also show that any expression achieving the target \emph{must} take the form above. 
We restrict the set of possible forms by (1) inducting to show that each subtree of a solution must have degree in $x$ equal to its number of leaves of value $x$, (2) counting primes factors of the highest degree term to show that subtrees with no $x$ values must be of form $\{\pm b_i, \pm c_i, \pm b_i\pm c_i\}$, (3) a divisibility argument to show that sums of elements of form $\{\pm b_i, \pm c_i, \pm b_i\pm c_i\}$ as appearing in any evaluation of a subtree is nonzero, and (4) an argument on the degree of $y$ for terms with degree 0 in $x$ to show that these sums can never be canceled.

The bulk of the proof will occur in Lemma~\ref{lem:leaves:plus-minus-times:3}, which states that all subtrees of a solution to this instance evaluate to monomials.
To prove Lemma~\ref{lem:leaves:plus-minus-times:3}, we use 3 helper lemmas, which we now proceed to state and prove:



\begin{lemma}
\label{lem:leaves:plus-minus-times:1}
Let $T_{full}$ be the tree representation of a solution to instance $I_A$. All subtrees $T$ of $T_{full}$ have an evaluation with degree in $x$ equal to the number of $x$ terms in its leaves. If two subtrees both have a nonzero number of $x$ terms in its leaves, their lowest common ancestor must be a $\times$ operator.
\end{lemma}

\begin{proof}[Proof of Lemma~\ref{lem:leaves:plus-minus-times:1}]
We prove by strong induction on subtrees of increasing height that subtrees $T$ of $T_{full}$ have an evaluation with degree in $x$ at most $k_T$, the number of $x$ terms in its leaves. We then use the equality condition of the induction to show that the evaluation has degree in $x$ exactly $k_T$.

It is clear that all leaves, that is, subtrees of height 0, have degree in $x$ equal to 1 if it is an $x$ term and 0 if it is an $a_i y$ term.

Assume that the degree in $x$ of the evaluation of all subtrees of height at most $h$ is at most the number of leaves of value $x$. We show this is also true for all subtrees of height $h+1$.

Let $T$ be a subtree of height $h+1$ with left subtree $L$ with $k_L$ leaves of value $x$ and right subtree $R$ with $k_R$ leaves of value $x$. Since $L$ and $R$ have height at most $h$, the degree in $x$ of their evaluations are $k_L$ and $k_R$. If $T$ is rooted at an $+$ or $-$ operator, then the degree in $x$ of the evaluation of $T$ is at most $\max(k_L, k_R)$. If $T$ is rooted at a $\times$ operator, then the degree in $x$ of the evaluation of $T$ is at most $k_L + k_R$. Thus the evaluation of $T$ has degree in $x$ at most $k_T = k_L + k_R$, the number of leaves of value $x$.

Note that the evaluation of $T$ has $x$-degree exactly $k_T = k_L + k_R$ when the evaluations of $L$ and $R$ have $x$-degree exactly $k_L$ and $k_R$ (since they have degree in $x$ at most $k_L$ and $k_R$). If $k_L$ and $k_R$ nonzero, then in order for $T$ to have $x$-degree $k_T$, it must be the case that $L$ and $R$ are multiplied together. Since any $T_{full}$ evaluates to $t(x,y) = L^n y^n x^{n-1} \prod_{i \in [n]} p_i$, which has $x$-degree equal to the number of $x$ leaves, propagating this property from the root of the tree to the leaves we find that any subtree $T$ has degree in $x$ equal to the number of leaves of value $x$.
\end{proof}

\begin{lemma} \label{lem:leaves:plus-minus-times:bici}
Let $b_i$ and $c_i$ be as defined above.  Let $T_{full}$ be the tree representation of a solution to $I_A$.  All subtrees of $T_{full}$ whose evaluation has degree 0 in $x$ must evaluate to $\pm b_i$, $ \pm c_i$, or $ \pm b_i\pm c_i$. In other words, the subtree  $ \pm b_i\times c_i$ cannot exist.
\end{lemma}

\begin{proof}[Proof of Lemma~\ref{lem:leaves:plus-minus-times:bici}]
First, we show that the coefficient of the term with highest degree in $x$ in the evaluation of $T_{full}$ must be of the form:
$$\prod_{i=1}^n F_i$$ 
where $F_i  \in \{\pm 1, \pm b_i, \pm c_i, \pm b_i\pm c_i, \pm b_ic_i\}$ for all $i \in [n]$. 

Let $n_j$ be the $j$th leaf from the left in leaf ordering.  Consider the subtree rooted at highest ancestor of $n_j$ not containing $n_{j_\ell}$, $n_{j_r}$ with $j_\ell < j < j_r$. The leaves of this subtree are within the range $[j_\ell+1, j_r-1]$. 
Then any term with degree 0 in $x$ formed with $b_i$ or $c_i$ will in the form $F_i$ as above, before being operated with a subtree containing an $x$. This term only contribute the coefficient of the term with highest degree in $x$ if it is multiplied.

Recall that the target $t(x,y) = L^n y^n x^{n-1} \prod_{i \in [n]} p_i$ contains $n$ unique primes $p_1, \ldots, p_n$.
Furthermore, recall that by construction $b_i$ and $c_i$ both contain $p_i$ as a factor.
Thus, a solution $T_{full}$ must have $F_i \in \{\pm b_i, \pm c_i, \pm b_i\pm c_i\}$ for all $i$ because there is exactly one factor of each $p_i$ in $t(x,y)$. We may never multiply $b_i$ and $c_i$, as $b_i c_i$ has 2 factors of $p_i$, so if it were multiplied into a subtree containing $x$, the resulting evaluation would have too many $p_i$ factors to meet the target.
\end{proof}

\begin{lemma} \label{lem:leaves:plus-minus-times:d}
Let $b_i$ and $c_i$ be as defined above.
All sums of the following form, where at least one $d_j$ is nonzero, have nonzero evaluations: \begin{align}
d_1b_i +d_2c_i +d_3b_{i'}+ d_4c_{i'}, \ \  d_j\in \{0,\pm 1\}. \label{eqn:d-eqn}
\end{align}
\end{lemma}

\begin{proof}[Proof of Lemma~\ref{lem:leaves:plus-minus-times:d}]
Suppose, for contradiction, that we do have an expression of the form above which evaluates to $0$. 
\begin{multline}
d_1b_i +d_2c_i +d_3b_{i'}+ d_4c_{i'} \\ = \frac{Ly}{2}\left(d_1p_i\left(a_i+\frac{1}{a_i}\right)+ d_2p_i\left(a_i - \frac{1}{a_i}\right)+ d_3p_{i'}\left(a_{i'}+\frac{1}{a_{i'}}\right) +d_4p_{i'}\left(a_{i'}-\frac{1}{a_{i'}}\right)\right) \\ = 0 \end{multline}

We can cancel out the $\frac{Ly}{2a_i a_{i'}} $ factor in our equation:
\begin{align}
    (d_1 + d_2) p_i a_i^2 a_{i'}  + (d_3 + d_4) p_{i'} a_i a_{i'}^2 +
    (d_1-d_2)p_i a_{i'} + (d_3 - d_4)p_{i'}a_{i}= 0\label{eqn:intseqn}
\end{align}

For Equation~\ref{eqn:intseqn} to hold, we require $(d_1 - d_2) p_i a_{i'} \equiv 0 \mod a_{i}$ and $(d_2 - d_3) p_{i'} a_{i} \equiv 0 \mod a_{i'}$. Since these $p_i$, $a_{i'}$, and $a_{i}$ are distinct, there is a prime factor, which we call $q_i$ of $a_i$ that is not in either $a_{i'}$ or $p_i$, and by construction $q_i > 3$. Since, by construction, $d_1 - d_2 \in \{0,\pm 1,\pm 2\}$ and $d_1 - d_2 \equiv 0 \mod q_i$, we must have $d_1 = d_2$. Similarly, $d_3 =d_4$. 

Applying these two substitutions, Equation~\ref{eqn:intseqn} becomes $2d_1p_i a_i^2 a_{i'}  + 2d_3 p_{i'} a_i a_{i'}^2 = 0$. Now considering this equation modulo $p_i$ and $p_{i'}$ in the same way as before, we find we must have $d_3, d_1 = 0$. This implies all $d_j = 0$, a contradiction.
\end{proof}

\begin{lemma}
\label{lem:leaves:plus-minus-times:3}
For a solution $T_{full}$ to $I_A$ as defined above, consider any subtree $T$ containing $x$. Its evaluation $\ev(T)$ is a monomial in $x$ and $y$.
\end{lemma}

\begin{proof}[Proof of Lemma~\ref{lem:leaves:plus-minus-times:3}]
We show by an induction from the root that the evaluation of any  subtree $T$ of solution $T_{full}$, such that $T$ contains at least one $x$ leaf, is a monomial in $x$ and $y$.

At the root of the tree, our target fulfills this condition.

Now we will show that if some $x$-containing subtree $T$ evaluates to a monomial, then its left subtree $L$ and right subtree $R$ do as well. There are two cases:

Case 1: $L$ and $R$ contain at least one leaf with value $x$.
By Lemma~\ref{lem:leaves:plus-minus-times:1}, we know that the operator connecting $L$ and $R$ must be $\times$.  Suppose by way of contradiction that the evaluation of either $L$ or $R$ had more than one term.  Then their product would have more than one term.  But $T$ is a monomial.  Thus, $L$ and $R$ must both be monomials in $x$ and $y$.

Case 2: One of $L$ or $R$ contain at least one leaf with value $x$ and the other does not.
We will show that in this case, for the evaluation of $T$ to be a monomial it must be rooted at a $\times$ operator.  This will show that the evaluations of both $L$ and $R$ must be monomials. which will complete our proof. 

Assume for the sake of contradiction there is some subtree of $T_{full}$ such this subtree $T$ is rooted at a $+$ or $-$ operator, it evaluates to a monomial in $x$ and $y$, and it has one child with leaves of value $x$ and another without. 

Case (i): $T$ has no $\times$ operator.
Then by Lemma~\ref{lem:leaves:plus-minus-times:1} it only has one $x$ leaf and its evaluation is $\pm x$ plus a sum with form as in Equation~\ref{eqn:d-eqn}.  By  Lemma~\ref{lem:leaves:plus-minus-times:d} the sum is nonzero, so $T$ would be a binomial, a contradiction. 

Case(ii): If $T$ has a $\times$ operator, but not at the root.
Let $A$ be the subtree of highest $\times$ operator node in $T$.  By Lemma~\ref{lem:leaves:plus-minus-times:bici}, $A$ must contain some $x$ leaf, and by Lemma~\ref{lem:leaves:plus-minus-times:1} $A$ must contain all $x$ leaves in $T$. Thus, $T$ must equal $A$ plus or minus multiples of $b_i$ and $c_i$ (but not $b_i \times c_i$, by Lemma~\ref{lem:leaves:plus-minus-times:bici}). 
Observe that the sum that we add or subtract from $A$ to form $T$ has degree 0 in $x$, degree 1 in $y$, and by Lemma~\ref{lem:leaves:plus-minus-times:d} it is nonzero. However, the evaluation of $A$ has no term with degree 0 in $x$ and degree 1 in $y$.  This is because $A$ was defined to have a $\times$ node at the root, so any of its terms with degree 0 in $x$ must be a product of two terms with degree 0 in $x$, one from each of its subtrees.  This would require $A$ to have degree at least 2 in $y$, which it does not.  Thus $T$ cannot be a monomial if it is rooted at a $+$ or $-$ operator.
\end{proof}

We are now ready to prove the claim.

\begin{proof}[Proof of Claim 
\ref{claim:EL-plus-minus-times}
]
Let $A$ be an instance of \SetProductPartitionBound{3}, and let $I_A$ be the $\AECEL{\nats[x,y]}{\{+,-,\times\}}$ instance as defined above.

As mentioned in the sketch, if an instance of this product partition variant is solvable, then the constructed instance evaluates to $t(x,y) = L^n y^n x^{n-1} \prod_{i \in [n]} p_i$ when we have $(b_i + c_i)$ for $a_i$ in one partition and $(b_i - c_i)$ for $a_i$ in the other, and the $\times$ operator at every other node.
The partition corresponds to whether the $a_i$ was written as a difference or a sum.

We now show that any expression achieving the target \emph{must} take the form above, and thus implies the existence of a solution to the instance of the product partition variant. 

By Lemma \ref{lem:leaves:plus-minus-times:3}, all subtrees of a solution evaluate to monomials. Thus, addition and subtraction can only occur between $b_i$ and $c_i$. Further, because our target has only one factor of each $p_i$, addition and subtraction \emph{must} occur between $b_i$ and $c_i$, else by Lemma~\ref{lem:leaves:plus-minus-times:3} both would need to be multiplied in, resulting in an evaluation at the root with too many factors of $p_i$. Thus any tree evaluating to the target must be the product of $x$ leaves and the sums or differences of $b_i$ and $c_i$.
This completes the proof.
\end{proof}

\subsection{Weak NP-hardness of AEC Enforced Leaves \texorpdfstring{$\{-, \times\}$}{\{−, ×\}}}
\label{appendix-leaves:minus-times}
\begin{claim}

$\AECEL{\nats[x,y]}{\{-,\times\}}$ is weakly NP-hard.
\label{appendix-leaves:MTclaim}
\end{claim}
The proof for this claim is very similar to the proof for AEC Enforced Leaves $\{+,-,\times\}$ in Section~\ref{appendix-leaves:plus-minus-times}, and will use lemmas and modifications of lemma from the previous section. The key difference is that we use an additional leaf to represent addition using only subtraction and brackets, as $a - (0 - b) = a+b$. 

As in Section~\ref{appendix-leaves:plus-minus-times}, let $A = \{a_1,\ldots, a_n\}$ be an instance of \SetProductPartitionBound{3}, let $L = 2 \prod_{i \in [n]} a_i$, and let $p_j$ for $j \in [n]$ be unique primes greater than 3 that are coprime to $\prod_{i \in [n]} a_i (a_i^2 + 1)(a_i^2 - 1)$. Similar to Section~\ref{appendix-leaves:plus-minus-times}, let $b_i$ and $c_i$ be $y$ times what it was before, so that $b_i + c_i = (La_i)p_iy^2$ and $b_i - c_i = (L/a_i)p_i y^2$.
We cannot use 0 in the equation, since it is non-positive.  Instead, we will add some terms of $b'_i = b_i + p_i y$.

Let instance $I_A$ of $\AECEL{\nats[x,y]}{\{-,\times\}}$ 
enforce the following order of leaves:
$$
b'_1  ~\enspace~ p_1 y  ~\enspace~ c_1  ~\enspace~ x ~\enspace~ b'_2  ~\enspace~ p_2 y ~\enspace~ c_2  ~\enspace~ x~ \enspace~ \cdots ~\enspace~ x~ \enspace~ b'_n  ~\enspace~ p_n y ~ \enspace~ c_n,
$$
and have target polynomial $t(x,y) = L^n y^{2n} x^{n-1} \prod_{i \in [n]} p_i$.

The structure of this proof is similar to the previous section. We can directly use the statement and proof of Lemma~\ref{lem:leaves:plus-minus-times:1} and Lemma~\ref{lem:leaves:plus-minus-times:d} (with an extra factor of $y$ canceled in the proof of the latter). We now state the analogue of Lemma~\ref{lem:leaves:plus-minus-times:bici}, as well as the analogue of Lemma~\ref{lem:leaves:plus-minus-times:3} (which is the same statement but has an altered proof due to the new problem instance structure), show how the latter lemma can be applied to prove Claim~\ref{appendix-leaves:MTclaim}, and finally prove the lemmas.

\begin{lemma} \label{lem:leaves:minus-times:bici}
In the evaluation of a solution $T_{full}$ to $I_A$, the coefficient of the term with highest degree in $x$ must be of the form:
$$\prod_{i} F_i(b'_i, p_i y, c_i)$$ 
where $F_i(b'_i, p_i y, c_i) \in \{c_i, b_i, b_i\pm c_i\}$.
\end{lemma}

\begin{lemma}
\label{lem:leaves:minus-times:3}
For a solution $T_{full}$ to $I_A$ consider any subtree $T$ containing $x$. Its evaluation $\ev(T)$ is a monomial in $x$ and $y$.
\end{lemma}

To complete the proof of this claim, we prove the two modified lemmas.

\begin{proof}[Proof of Lemma~\ref{lem:leaves:minus-times:bici}]
As in the proof of Lemma~\ref{lem:leaves:plus-minus-times:bici}, the coefficient of the term with highest degree in $x$ in the evaluation of $T_{full}$ must be a product of functions of $b'_i$, $p_i y$, and $c_i$. Since there is exactly one factor of each $p_i$ in the target, multiplication cannot be used in the construction of $F_i(b'_i, p_i y, c_i)$. Thus using 0, 1, or 2 subtraction operations, we may form the set $\{b'_i, c_i, b_i, p_i y - c_i, b_i\pm c_i\}$. Since the target is a monomial in $x$ and $y$, then each $F_i(b'_i, p_i y, c_i)$ must also be a monomial in $x$ and $y$. Thus $F_i(b'_i, p_i y, c_i) \in \{c_i, b_i, b_i\pm c_i\}$.
\end{proof}

\begin{proof}[Proof of Lemma~\ref{lem:leaves:minus-times:3}]
As in the proof of Lemma~\ref{lem:leaves:plus-minus-times:3}, we show by an induction from the root that for any solution $T_{full}$,  the evaluation of any subtree $T$ containing $x$ is a monomial in $x$ and $y$.

As before, our target fulfills this condition at the root of the tree, and if both the left subtree $L$ and right subtree $R$  contain at least one leaf with value $x$, then the evaluation of both subtrees are monomials in $x$ and $y$. To complete the proof of this lemma we show that the evaluation of $T$ to be a monomial in $x$ and $y$, then its left or right subtree containing $x$ must also be a monomial. If $T$ is rooted at a $\times$ operator and it evaluates to a monomial then the evaluations of both $L$ and $R$ must be monomials.

Unlike before, Lemma~\ref{lem:leaves:minus-times:bici} is weaker than its analogue Lemma~\ref{lem:leaves:plus-minus-times:bici}, so there is a more complex analysis of the case where $T$ is rooted at a $-$ operator, one of $L$ or $R$ contain at least one leaf with value $x$, and the other does not. 

Assume for the sake of contradiction there is some subtree of $T_{full}$ such this subtree $T$ is rooted at a $-$ operator, evaluates to a monomial in $x$ and $y$, and it has one child without leaves of value $x$ and another which evaluates to a monomial in $x$ and $y$. 

Let $A$ be the subtree rooted at the highest $\times$ operator that is an ancestor to some $x$ leaf in $T$, if there is such an operator.  If there is no such operator, then by Lemma~\ref{lem:leaves:plus-minus-times:1}, it $A$ contains only one $x$ leaf. Let $A$ be the leaf $x$. If there is such an operator, then by Lemma~\ref{lem:leaves:plus-minus-times:1}  $A$ contains all $x$ leaves in $T$. There are at most 6 leaves in $T$ not in $A$, the 3 to the left of the leftmost $x$ in $A$ which we assign index $i$, and the 3 to the right of the rightmost $x$ in $A$ which we assign index $j$. There is at least one leaf in $T$ not in $A$, as $T$ is rooted at a $-$ operator and $A$ lies within it and is rooted at a $\times$ operator.

We use the following properties:

\textbf{Property 1. } $\ev(T) - \ev(A)$ has degree 0 in $x$ and degree 2 or 3 in $y$.

By Lemma~\ref{lem:leaves:plus-minus-times:1}, there are no $x$ leaves in $T$ not in $A$, so $\ev(T) - \ev(A)$ must have degree 0 in $x$.  By Lemma~\ref{lem:leaves:minus-times:bici}, at least one of $b'_i$ and $c_i$ must be in $F_i(b'_i, p_i y, c_i)$, so the degree in $y$ of $\ev(T) - \ev(A)$ is at most 3, which occurs when it is the product of $p_i y$ and $b'_i$ or $c_i$. Since there is at least one leaf in $T$ not in $A$, then $\ev(T) - \ev(A)$ has degree at least 1 in $y$. Further, if there is just a single leaf it cannot be $p_i y$, as that would require multiplying both $b'_i$ and $c_i$ into subtrees containing containing $x$, which would put too many $p_i$ terms in the target. Thus $\ev(T) - \ev(A)$ has degree at least 2 in $y$.

\textbf{Property 2. } $\ev(T) - \ev(A)$ is nonzero.

If there are no $\times$ operators in $T$ not in $A$, then the evaluation of $T$, $\ev(T)$, has form $d_1 b'_i + d_5 p_i y d_2 c_i \pm \ev(A) + d_3 b'_{j} + d_6 p_{j}y + d_4 c_{j}$ for $d_i \in \{0, \pm 1\}$ and not all $d_i$ zero. Note that $\ev(T) - \ev(A)$ has degree 0 in $x$ and degree 1 or 2 in $y$. Observing the part with degree 0 in $x$ and degree 2 in $y$, they are a sum with form as in Equation~\ref{eqn:d-eqn}, which by Lemma~\ref{lem:leaves:plus-minus-times:d} is nonzero. 

If there is some $\times$ operator in $T$ not in $A$, then by Lemma~\ref{lem:leaves:minus-times:bici} the only $F_i(b'_i, p_i y, c_i)$ using just less than two leaves is $c_i$, so the evaluation at the $\times$ operator must be $b'_i p_i y$ (or $b'_j p_j y$). The evaluation of these subtree that is degree 3 in $y$ is $b_i p_i y$ and $b_j p_j y$. Both these values are nonzero, as are their sum and difference. Thus $\ev(T) - \ev(A)$ is again nonzero.

Now, we can finally prove that $\ev(T)$ is not a monomial under these conditions, a contradiction.

If $\ev(A)$ is a monomial in $x$ and $y$, then it cannot have any part with degree 0 in $x$. By the two properties,  $\ev(T) - \ev(A)$ is nonzero and has degree 0 in $x$, so their sum $\ev(T) = (\ev(T) - \ev(A)) + \ev(A)$ is not a monomial.

If $\ev(A)$ is not a monomial, then we let $A_L$ and $A_R$ be the left and right subtrees of $A$. If at least one of $\ev(A_L)$ and $\ev(A_R)$ has no part with degree 0 in $x$, then $\ev(A)$ has no terms with degree 0 in $x$. As before since $\ev(T) - \ev(A)$ is nonzero and has degree 0 in $x$, then $\ev(T)$ is not a monomial. If both $\ev(A_L)$ and $\ev(A_R)$ have some part with degree 0 in $x$, then these parts must also have degree at least 2 in $y$ (by an analogue of the argument in the first property). Thus the degree in $y$ of the part of $\ev(A)$ with degree 0 in $x$ is at least 4. By the first property  $\ev(T) - \ev(A)$ has degree at most 3 in $y$, so the parts with with degree 0 in $x$ in $\ev(A)$ and $\ev(T) - \ev(A)$ cannot cancel, implying $\ev(T)$ is not a monomial.
\end{proof}

\begin{proof}[Proof of Claim~\ref{appendix-leaves:MTclaim}]

The idea of this proof is the same as the proof from Section~\ref{appendix-leaves:plus-minus-times}.  However, this time, we will construct either $((b'_i-p_i y) -c_i) = b_i-c_i$ or $(b'_i-(p_i y -c_i)) = b_i+c_i$, making the placement of parentheses indicate the partition membership instead of the sign.

Like before, Lemma~\ref{lem:leaves:minus-times:3} implies that subtraction can only occur between terms with degree 0 in $x$. Further, because our target has only one factor of each $p_i$, subtraction \emph{must} occur between terms with degree 0 in $x$, else by Lemma~\ref{lem:leaves:plus-minus-times:3} both would need to be multiplied in, resulting in an evaluation at the root with too many factors of $p_i$. Thus any tree evaluating to the target must be the product of $x$ leaves and one of the two subtraction structures $((b'_i-p_i y) -c_i) = b_i-c_i$ or $(b'_i-(p_i y -c_i)) = b_i+c_i$, on the terms of degree 0 in $x$ as we desire.

This shows that any solution to the instance can be read as a solution to the underlying \SetProductPartitionBound{3} instance, completing the proof.
\end{proof}



\subsection{Weak NP-hardness of Enforced Leaves \texorpdfstring{$\{+, \times, \div\}$}{\{+, ×, ÷\}}}
\label{leaves:plus-times-div}

\begin{claim}
\label{appendix-leaves:PTD}
$\AECEL{\nats[x]}{\{+,\times,\div\}}$ is weakly NP-hard.
\end{claim}

\begin{proof}
We will reduce from \ProductPartition to $\AECEL{\nats[x]}{\{+,\times,\div\}}$. For an instance of \ProductPartition with elements $\{a_1, \ldots, a_n\}$, we construct an instance of $\AECEL{\nats[x]}{\{+, \times, \div\}}$ with leaf order 
$$ a_1\ \ x\ \ a_2\ \ x\ \ \cdots\ \ x\ \ a_{n-1}\ \ x\ \ a_n $$  and target $t(x) = x^{n-1}$.

If the \ProductPartition instance is solvable then this $\AECEL{\nats[x]}{\{+,\times,\div\}}$ is solvable with operations $\times$ and $\div$ by multiplying the $x$s and the $a_i$s in the same partition as $a_1$ and dividing the other $a_i$s.

Furthermore, if an AEC-EL solution \emph{only} uses the $\times$ and $\div$ operators, observe that it yields a valid solution to $\ProductPartition$, since the numerator and denominator will provide two sets of $a_i$ with equal product.

We proceed to show that only $\times$ and $\div$ can be used in our constructed AEC instance.  This proof proceeds in four steps, which ultimately show that any evaluation subtree containing $k$ $x$ leaves can only be an $x^k$ term or an $x^{-k}$ term.
This implies that addition can never occur, as the alternating $x$ and $a_i$ pattern we constructed for leaf order would under addition create polynomials of more than one term, a contradiction.

Recall from the paired model of rational function computation in  Section \ref{sec:rational} that we may represent the evaluation of a subtree $T$ as a pair of integer polynomials $(f_T,g_T)$. Our proof proceeds in four steps:
\begin{enumerate}
    \item First we use an inductive argument to bound the degree of $f_T$ and $g_T$ at $k$, the number of $x$ leaves within $T$.
    \item Then, we note that the equality condition of the induction, the target is achieved only when one of $f_T$ or $g_T$ have degree exactly $k$. 
    \item We then show that the other integer polynomial ($g_T$ or $f_T$ respectively) is a constant.
    \item Finally, we show by induction from the root that both integer polynomials are monomials as desired.
\end{enumerate}
Note the first two steps are together the rational equivalent of Lemma~\ref{lem:leaves:plus-minus-times:1}.\newline

\textbf{Step 1.} \emph{For any subtree $T$ represented as $(f_T, g_T)$, the evaluation of $T$ has degree at most the number of $x$ leaves within $T$.}

We prove the first step by strong induction on subtrees of increasing height.  

All subtrees of height 0 (leaves), can be represented by the pair $(a_i,1)$
or $(x,1)$.

Assume that the evaluations of all subtrees of height at most $h$ are integer polynomials with degree at most the number of leaves of value $x$. We will show this is also true for all subtrees of height $h+1$.

Let $T$ be a subtree of height $h+1$ with left subtree $L$ with $k_L$ leaves of value $x$ and right subtree $R$ with $k_R$ leaves of value $x$. We will show that the evaluation of $T$ has integer polynomials with degree at most $k_T = k_L + k_R$, the number of leaves of value $x$.

Because $L$ and $R$ have height at most $h$, by the inductive hypothesis the degrees of $f_L$ and $g_L$ are at most $k_L$; similarly the degrees of $f_R$ and $g_R$ are at most $k_R$. 
Recall the definitions of the following operations:
\begin{enumerate}
    \item $(f_T,g_T) = (f_L,g_L) + (f_R,g_R) = (f_L g_R + g_L f_R,\ g_L g_R)$
    \item $(f_T,g_T) = (f_L,g_L) \times (f_R,g_R) = (f_L f_R,\ g_L g_R)$
    \item $(f_T,g_T) = (f_L,g_L) \div (f_R,g_R) = (f_L g_R,\ g_L f_R)$
\end{enumerate}
Thus,  the degrees of both $f_T$ and $g_T$ are at most the max of the degrees of $f_L f_R, f_L g_R, g_L f_R$, or $g_L g_R$.  Each of these terms is the product of two polynomials with degrees $k_L$ and $k_R$.  Thus, the degree of each of these is at most $k_L + k_R = k_T$.  This completes step 1.\newline

\textbf{Step 2.} \emph{Consider an AEC-EL solution tree $T_{full}$ that meets the target.  For a subtree $T$ of $T_{full}$, represented by $(f_T, g_T)$, one of $f_T$ or $g_T$ has degree $k_T$.}


Let $T$ be a subtree with child subtrees $L$ and $R$, and define $f_T, g_T, f_L, g_L, and f_R, g_R$ appropriately.
Let $h_T$ be either $f_T$ or $g_T$, whichever has larger degree. Define $h_L$ and $h_R$  in the same manner.
In this step, our goal is to show that $\deg(h_T) = k_T$ for any subtree of a valid AEC-EL solution.

Recall from Step 1 that the maximum of the degrees of $f_L f_R, f_L g_R, g_L f_R, g_L g_R$ is at most $k_T = k_L + k_R$.
Observe that equality holds only if $\deg(h_L) = k_L$ and $\deg(h_R) = k_R$; else the degree of the product is less than $k_T$.

Let $T_{full}$ be the full tree. Observe that if $T_{full}$ evaluates to the target polynomial, then $\deg(h_{T_{full}}) = k_{T_{full}} = n-1$.
Since this equality holds only if the equality condition also holds for its left and right subtree, induction from the root shows that one of $\deg(f_T) = k_T$ or $\deg(g_T) = k_T$ for all subtrees.  This proves the second step.\newline


\textbf{Step 3.} \emph{Consider an AEC-EL solution tree $T_{full}$ that meets the target.  For a subtree $T$ of $T_{full}$, one of $f_T$ or $g_T$ is constant.}

For all subtrees $T$, let $h_T$ be defined as in the previous step, and recall that it has degree $k_T$.  Let $c_T$ be the other polynomial and let its degree be $d_T$. 

Table \ref{tab:fg1} shows casework for determining the degrees of $f_T$ and $g_T$ given the options for $h_L$ and $h_R$, and the operation.
Observe that when combining $L$ and $R$ into $T$, for these three operations 
it always holds that 
$\deg(c_L) + \deg(c_R) \le \deg(c_T)$.


The evaluation of $T_{full}$ is some constant multiple of $(x^{n-1},1)$ so $\deg(c_{T_{full}}) = 0$. Since degrees are non-negative, the degrees of the $c$ polynomials are non-decreasing, and $\deg(c_{T_{full}}) = 0$, we find that $\deg(c_{T}) = 0$ for all subtrees $T$.


\begin{table}[h] 
\centering
\begin{tabular}{c|ccc}
              & \multicolumn{3}{c}{\textbf{Degrees of $f_T; g_T$ under operation}} \\
\textbf{$h_L, h_R$}
& $+$                                           & $\times$                         & $\div$                           \\ \hline
$f_L,f_R$          & $\max(k_L + d_R,d_L+k_R) ;\ d_L + d_R$ & \cellcolor[HTML]{FFEFD5}$k_L + k_R;\ d_L + d_R$   & $k_L + d_R;\ d_L+ k_R$     \\
$f_L,g_R$          & $k_L+k_R;\ d_L+k_R$            & $k_L + d_R;\ d_L+ k_R$                       & \cellcolor[HTML]{FFEFD5}$k_L + k_R;\ d_L + d_R$ \\
$g_L,f_R$          & $k_L+k_R;\ k_L+d_R$            & $d_L+k_R;\ k_L+d_R$                       & \cellcolor[HTML]{FFEFD5}$d_L + d_R;\ k_L + k_R$ \\
$g_L,g_R$          & $\max(k_L + d_R,d_L+k_R) ;\ k_L+k_R$ &\cellcolor[HTML]{FFEFD5}$d_L + d_R;\ k_L + k_R$   & $d_L+k_R;\ k_L+d_R$    
\end{tabular}
\caption{The degrees of $f_T; g_T$ under the three operations $+$, $\times$, and $\div$ in the cases $i,j$ for $i,j \in \{f,g\}$ where $\deg(i_L) = k_L$,  $\deg(j_R) = k_R$, and the other two have degrees $d_L$ and $d_R$. Cells where $\deg(h_T) = k_T$ and $\deg(c_T) = 0$ when $k_L + k_R = k_T$ and $d_L + d_R = 0$ are highlighted. }\label{tab:fg1}
\end{table}

\textbf{Step 4.} \emph{Let $T_{full}$ be an AEC-EL solution tree that meets the target.  For a subtree $T$ of $T_{full}$ represented by $(f_T,g_T)$, both $f_T$ and $g_T$ are monomials.}

We show this by induction from the root. 
At the root of the tree, we have 
$(f_{T_{full}}, g_{T_{full}}) = (x^{n-1}, 1)$, 
fulfilling our desired condition.

Now we will show that if some subtree $T$ fulfills $(f_T, g_T) = (b_n x^{k_T}, b_d)$ or $(f_T, g_T) = (b_n, b_d x^{k_T})$ for constant $b_n, b_d$, then its left subtree $L$ and right subtree $R$ 
also fulfill one of these two conditions.

Case 1: $k_L, k_R > 0$.  ($L$ and $R$ both contain at least one $x$ leaf.)

When $k_L, k_R > 0$, then $k_L, k_R < k_L + k_R = k_T$.
We know by the previous steps that one of $f_T$ and $g_T$ must have degree $k_T$ (call this $h_T$), and the other must have degree 0 (call this $c_T$).  Table \ref{tab:fg1} highlights the four situations where $\deg(h_T) = k_T$ and $c_T$ constant.  In particular, these situations only occur under the $\times$ and $\div$ operations.  The resulting expression is of the form: $(h_L h_R, c_L c_R)$ or $(c_L c_R, h_L h_R)$. 

Since $h_T$ and $c_T$ are monomials, then this implies $h_L h_R$ and  $c_L c_R$ are monomials. The latter is certainly true when $c_L$ and $c_R$ are constant, and the former is only true when both $h_L$ and $h_R$ are monomials, because if either has more than one term, then their product must have more than one term.
With the conditions from the previous steps, this means in particular that $(f_T, g_T) = (b_n x^{k_T}, b_d)$ or $(f_T, g_T) = (b_n, b_d x^{k_T})$ where $b_n$ and $b_d$ are constant.

Case 2: $k_R > 0$, $k_L = 0$.  ($R$ contains an $x$ leaf but $L$ does not.)

In this case, $L$ is an $a_i$ leaf, and both its integer polynomials are constant. We know $h_T, c_T, h_L, c_L, c_R$ are monomials, and we would like to show that $h_R$ is also a monomial.

Note that the operation at $T$ cannot be $+$ because $f_T = f_L g_R + g_L f_R$ is the sum of a degree $k_R$ polynomial and a constant, and thus cannot be a monomial.
If the operation at $T$ is $\times$ or $\div$, then $h_R$ is a monomial because if it were not, then $h_T$ could not be a monomial of degree $k_T$.

Case 3: $k_L > 0$, $k_R = 0$.  ($L$ contains an $x$ leaf but $R$ does not.)
This holds for the same reason as case 2.

Case 4: $k_L = k_R = 0$ never happens due to the ordering of our leaves.

This completes our proof.
\end{proof}

\section{Additional Results from Other Reductions}
\label{basic}
\label{standard:plus}
\label{standard:times}
\label{ops:plus}
\label{ops:times}
\label{leaves:plus}
\label{leaves:times}


In this section, we present some basic results in Arithmetic Expression Construction that do not use the rational framework: a classification of the hardness of $\{+\}, \{\times\},\{-\}, \{\div\}, \{+,-\},$ and $\{\times,\div\}$ Arithmetic Expression Construction for all variants, and of $\{-, \times\}$ and $\{+, -, \times\}$ Arithmetic Expression Construction for the Enforced Operations variant.

\SG{\mathbb{N}}{\{+\}}{\variant} and \SG{\mathbb{N}}{\{\times\}}{\variant} are trivial for all variants as any expression involving just an addition (multiplication) reduces to a sum (product) of the terms. Checking the solvability reduces to testing that $\sum_i a_i = t$ or $\prod_i a_i = t$ and is in P. 

For \SG{\mathbb{N}}{\{+, -\}}{\variant} and \SG{\mathbb{N}}{\{-\}}{\variant} variants we prove weak NP-hardness by reductions from $\textsc{Partition}$ or $\textsc{Partition-}n/2$. We then give pseudopolynomial algorithms for all variants of \SG{\mathbb{N}}{\{+, -\}}{\variant} and \SG{\mathbb{N}}{\{-\}}{\variant}.

\subsection{\texorpdfstring{$\{+, - \}$}{\{+, −\}}, \texorpdfstring{$\{ - \}$}{\{−\}} is Weakly NP-hard for All Variants}

\subsubsection{\texorpdfstring{$\{+, -\}$}{\{+, −\}} is Weakly NP-hard for All Variants}
\label{standard:plus-minus}
\label{ops:plus-minus}
\label{leaves:plus-minus}


To prove hardness of the Standard and Enforced Leaves variants, we reduce from \Partition.  Any expression $E$ using $\{+,-\}$ over the set $A$ to be partitioned can be written in the form 
$$\sum_{i \in \vert A \vert} c_ia_i = \sum_{c_i = 1}a_i - \sum_{c_i = -1} a_i$$
where $c_i\in \{\pm 1\}$ by distributing the sums and differences of $E$. 
We construct an instance of \SG{\mathbb{N}}{\{+, -\}}{\variant} given by $A$ with target $t = 0$. If there is a partition $(A_1, A_2)$ of $A$ with $\sum A_1 = \sum A_2$, then we can form the expression: 
$$a_{i_1} + a_{i_2} + \cdots + a_{i_k} - a_{j_1} - a_{j_2} -\cdots -a_{j_{k'}} = 0, a_{i} \in A_1, a_{j}\in A_2$$ 
and if there is an AEC solution we can recover a partition from the coefficients $c_i$.

For Enforced Leaves, we choose an arbitrary ordering of the partition set $A$ and produce a \SGEL{\mathbb{N}}{\{+, -\}} instance:
$$a_1\enspace a_2\enspace \cdots\enspace a_n, t = 0$$
The expressions formable with this ordering are:
$$a_1\pm a_2\pm a_3\pm \cdots\pm a_n$$
These are the same expressions in the Standard case except that $a_1$ must be positive.  In the context of our reduction from $\textsc{Partition}$ we can have $a_1\in A_1$ without loss of generality, so this does not affect the reduction. 

For Enforced Operations, we reduce from $\textsc{Partition-}n/2$.  We create an \SGEO{\mathbb{N}}{\{+, -\}} instance using the same set of integers $A$ and an expression which equates the difference of two sets of $n/2$ integers with target $t=0$: $$(\Box + \Box + \Box + \cdots) - (\Box + \Box + \Box + \cdots)$$
It is clear this expression is $0$ exactly when the subtracted terms are a partition of $A$.


\subsubsection{Standard \texorpdfstring{$\{-\}$}{\{−\}} is Weakly NP-hard}
\label{standard:minus}
Given an instance of $\textsc{Partition}$, we produce an instance of \SGSTD{\mathbb{N}}{\{-\}} with the same set of positive integers $A$ and target $t = 0$.
If the \Partition instance has a solution, we can construct an expression of the form $$(p_1 - n_2 - \cdots - n_{|A_2|}) - (n_1 - p_2 - \cdots - p_{|A_1|}) = \sum A_1 - \sum A_2 = 0$$ with $p_i$ in $A_1$ and $n_i$ in $A_2$. Conversely, any solution to the produced \SGSTD{\mathbb{N}}{\{-\}} instance can be factored into this form. Since $\sum A_1 = \sum A_2$ and $a_i > 0$ for all $a_i \in A_1 + A_2$, both $A_1$ and $A_2$ have nonzero sum and are thus nonempty.


\subsubsection{Enforced Operations \texorpdfstring{$\{-\}$}{\{−\}} is Weakly NP-hard}
\label{ops:minus}
 Given an instance of \PartitionEqual, we construct an instance of \SGEO{\mathbb{N}}{\{-\}} with the same set of positive integers $A$, target $t = 0$, and operation tree formula 
$$(\Box - \Box - \Box - \dots) - (\Box - \Box - \Box - \dots)$$
If this instance of $\{-\}$-\textsc{aec-eo} has a solution, we can divide $A$ into two complementary subsets $A_1, A_2$ as above in Section~\ref{standard:minus} such that $|A_1|=|A_2|=n/2$ and $\sum A_1 = \sum A_2$.

\subsubsection{Enforced Leaves \texorpdfstring{$\{-\}$}{\{−\}} is Weakly NP-hard} 
\label{leaves:minus}

Given an instance of $\textsc{Partition}$, 
we construct an \SGEL{\mathbb{N}}{\{-\}} instance with the same set of integers $a_i \in A$ plus one $1$, target $t = 1$, and leaf order
$$1\ \ a_1 \ \ a_2 \ \ a_3 \ \ \ldots \ \ a_n.$$

If our instance of $\textsc{Partition}$ has a solution such that $\sum A_1 = \sum A_2$ where $A_1, A_2$ are non-empty complementary subsets of $A$, there is a solution to the corresponding \SGEL{\mathbb{N}}{\{-\}} instance. Without loss of generality, let $A_1$ be the subset that contains $a_1$. We want our expression to simplify to $$1 - \sum A_1 + \sum A_2,$$
or equivalently
$$1 + \sum c_i a_i$$
where $c_i = -1$ for $a_i \in A_1$   and $c_i = 1$  for  $a_i \in A_2$.

In all constructions, $c_1 = -1$, but for any choice of $c_i$ where $i > 1$, we can construct a corresponding expression for our instance of \SGEL{\mathbb{N}}{\{-\}}. For a given term $a_i$, if we wish for $c_{i}$ to have the opposite sign of $c_{i-1}$, we insert a parentheses before $a_{i-1}$. We close the parenthetical expression at the end of the expression:
$$\dots\ -\ (a_{i-1}\ -\  a_{i}\ -\ a_{i+1} \dots\ a_{n})$$
The sign of $c_{i-1}$ is unchanged from before, as 
$\dots\ -\ a_{i-1}\ -\ \dots$ is equivalent to $\dots\ -\ (a_{i-1}\ -\ \dots)$ with respect to the sign of $c_{i-1}$. Thus we can construct an expression with the given leaf ordering such that $c_i=-1$ if $a_i \in A_1$ and $c_i=1$ if $a_i \in A_2$. This evaluates to $1$, our target.

Conversely if there is a solution for our instance of \SGEL{\mathbb{N}}{\{-\}},
$$1 + \sum c_i a_i = 1$$

Define $A_1 = \{a_i \mid a_i \in A, c_i = -1\} $ and $A_2 = \{a_i \mid a_i \in A, c_i = 1\} $.
$$\sum A_2 - \sum A_1 = \sum c_i a_i = 0$$
$$ \sum A_1 = \sum A_2$$
The two complementary subsets $A_1, A_2$ of $A$ have nonzero sum and are thus nonempty.



\subsection{\texorpdfstring{$\{\times, \div \}, \{\div \}$}{\{×, ÷\}, \{÷\}} is Strongly NP-hard for All Variants}

The reductions for \SG{\mathbb{N}}{\{\times, \div\}}{\variant} closely follow our reductions for \SG{\mathbb{N}}{\{+, -\}}{\variant}. We reduce from \ProductPartition or \ProductPartitionEqual, which are both NP-hard. 

\subsubsection{\texorpdfstring{$\{\times, \div \}$}{\{×, ÷\}} is Strongly NP-hard for All Variants}
\label{standard:times-div}
\label{leaves:times-div}
\label{ops:times-div}

To prove hardness of the Standard and Enforced Leaves variants, we reduce from \ProductPartition.  Given a \ProductPartition instance $A$, we will produce an instance of \SG{\mathbb{N}}{\{\times, \div\}}{\variant} with the same set of positive integers $A$ and target $t = 1$.

If this instance of \SG{\mathbb{N}}{\{\times, \div\}}{\variant} has a solution, we can divide $A$ into two complementary subsets $A_1$ and $A_2$ such that $\prod A_1 / \prod A_2 = t$. As $t=1$, $\Pi A_1 = \Pi A_2$. This means that $A_1$ and $A_2$ are a valid product-partitioning.
If the instance of \ProductPartition has a solution, we know two complementary subsets $A_1, A_2$ of $A$ exist such that $\Pi A_1 = \Pi A_2$. Our instance of \SG{\mathbb{N}}{\{\times, \div\}}{\variant} allows us assign nodes to produce the expression $\Pi A_1 / \Pi A_2 = 1$. 

We can construct a similar reduction for Enforced Leaves by choosing an arbitrary ordering of the elements in $A$.

For Enforced Operations, we reduce from \ProductPartitionEqual. We construct an instance of  \SGEO{\mathbb{N}}{\{\times, \div\}} using the same set of positive integers $A$, target $t=1$, and an expression tree of the form
$$(n_1 \times n_2 \times \cdots \times n_{n/2}) \div (d_1 \times d_2 \times \cdots \times d_{n/2})$$
where there are $n/2$ positive integers between each pair of parenthesis.

\subsubsection{Standard \texorpdfstring{$\{\div \}$}{\{÷\}} is Strongly NP-hard}
\label{standard:div}
Given an instance of \ProductPartition, we will produce an instance of \SGSTD{\mathbb{N}}{\{\div\}} with the same set of positive integers $A$ and target $t = 1$.


Our instance of \SGSTD{\mathbb{N}}{\{\div\}} allows us assign nodes to produce an expression equivalent to $\Pi A_1 \div \Pi A_2 = 1$:
$$ (n_1 \div d_2 \div \cdots \div d_{|A_2|}) \div (d_1 \div n_2 \div \cdots \div n_{|A_1|}) = \Pi A_1 \div \Pi A_2$$
where $A_1 + A_2 = A$;  $n_i$ in $A_1$ and  $d_i$ in $A_2$
If  $\Pi A_1 = \Pi A_2 \ne 1$, both $A_1$ and $A_2$ must be nonempty. Otherwise, $\Pi A_1 = \Pi A_2 = 1$ and all integers $a_i \in A$ are 1, making both problems trivial.

\subsubsection{Enforced Operations \texorpdfstring{$S \supseteq \{ \div \}$}{S ⊇ \{÷\}} is Strongly NP-hard}
\label{ops:div}
\label{ops:plus-div}
\label{ops:minus-div}
\label{ops:plus-minus-div}
\label{ops:plus-times-div}
\label{ops:minus-times-div}
\label{ops:plus-minus-times-div}
Given an instance of \ProductPartitionEqual, we will construct an instance of \SGEO{\mathbb{N}}{\{\div\}} with the same set of positive integers $A$, target $t = 1$, and tree requiring formula $$(n_1 \div d_2 \div \cdots \div d_{n/2}) \div (d_1 \div n_2 \div \cdots \div n_{n/2})$$ where there are $n/2$ positive integers between each pair of parenthesis. The proof of correctness follows the same method as used for \ref{ops:minus}.

Since Enforced Operations allows us to specify the expression tree, we can restrict expressions to $S \subset \{ \div \}$ in reductions to show that all set of operations $S \supseteq \{ \div \}$ are strongly NP-hard.

\subsubsection{Enforced Leaves \texorpdfstring{$\{\div \}$}{\{÷\}} is Strongly NP-hard}
\label{leaves:div}
Given an instance of \ProductPartition, we will produce an instance of \SGEL{\mathbb{N}}{\{\div\}} with the same set of positive integers $A$ and $n-1$ ones, target $t = 1$, and leaf order as $$a_1 \  \ 1 \ \ a_2 \ \ 1 \ \ a_3 \ \ 1  \ \ \ldots \ \ 1  \ \ a_n$$

If this instance of \SGEL{\mathbb{N}}{\{\div\}} has a solution, we expand our expression such that $\div (1 \div a_i)$ is simplified to $\times a_i$. We can then divide $A$ into two complementary subsets $A_1$, containing elements preceded by $\times$ after expansion, and $A_2$, containing elements preceded by $\div$ after expansion, so $\prod A_1 / \prod A_2 = t$.  As $t=1$, $\Pi A_1 = \Pi A_2$. 

If the instance of \ProductPartition has a solution, we know two complementary subsets $A_1, A_2$ of $A$ exist such that $\Pi A_1 = \Pi A_2$. Assume $a_1 \in A_1$. Our instance of \SGEL{\mathbb{N}}{\{\div\}} allows us assign operations to produce an expression equivalent to $\Pi A_1 / \Pi A_2 = 1$, by assigning $\div (1 \div a_i)$ equivalent to $\times a_i$ for $a_i \in A_1$ (except $a_1$, which is not preceded by anything and is already equivalent to $\times a_1$) and $\div 1 \div a_i$ equivalent to $\div a_i$  for $a_i \in A_2$. 

\subsection{Enforced Operations \texorpdfstring{$\{-,\times\}$}{\{−, ×\}} and \texorpdfstring{$\{+,-,\times\}$}{\{+, −, ×\}} are Strongly NP-hard}
\label{ops:minus-times}
\label{ops:plus-minus-times}
Given an instance of \ProductPartitionEqual, we will produce an instance of \SGEO{\mathbb{N}}{\{-, \times\}} with the same set of positive integers $A$, target $t = 0$, and tree requiring formula $$(\Box \times \Box \times \cdots \times \Box) - (\Box \times \Box \times \cdots \times \Box),$$ where there are  $n/2$ positive integers between each pair of parenthesis.

If the instance of \ProductPartitionEqual has a solution, we know two complementary subsets $A_1, A_2$ of $A$ exist such that $\Pi A_1 = \Pi A_2$. Our instance of \SGEO{\mathbb{N}}{\{-, \times\}} allows us assign nodes to produce the expression $\Pi A_1 - \Pi A_2 = 0$.

Similarly if this instance of \SGEO{\mathbb{N}}{\{-, \times\}} has a solution, we can divide $A$ into two complementary subsets $A_1$ and $A_2$ such that $|A_1|=|A_2|=n/2$ and  $\prod A_1 - \prod A_2 = t$. As $t=0$, $\Pi A_1 = \Pi A_2$.

We have shown \SGEO{\mathbb{N}}{\{-, \times\}} is strongly NP-hard. Since \SGEO{\mathbb{N}}{\{-, \times\}} is a special case of \SGEO{\mathbb{N}}{\{+, -, \times\}}, \SGEO{\mathbb{N}}{\{+, -, \times\}} is strongly NP-hard as well.

\subsection{Enforced Operations \texorpdfstring{$\{+,\times\}$}{\{+, ×\}} is weakly NP-hard}
\label{ops}
\label{ops:plus-times}


%
To prove that $\AECEO{\nats}{\{+,\times\}}$ is weakly NP-hard, we proceed by reduction from $3\textsc{-Partition}$-3, which is $3\textsc{-Partition}$ with the extra restriction that all the subsets have size 3.
Given an instance of $3\textsc{-Partition}$-3,  $A = \{a_1, a_2, \cdots, a_n\} $, construct instance $I_A$  of $\AECEO{\nats}{\{+,\times\}}$ with the same set of values $A$, target $t = \left( \frac{S}{n/3} \right)^{n/3}$, where $S = \sum_i a_i$, and expression-tree: 
$$(\Box + \Box + \Box) \times (\Box + \Box + \Box) \times \cdots \times (\Box + \Box + \Box),$$ 
where there are $n/3$ pairs of parentheses and $3$ positive integers between each pair of parentheses.

Given a solution of the $3\textsc{-Partition}$-3 instance, one can use the same partition to fill in the $3$-sums and solve our $\AECEO{\nats}{\{+,\times\}}$ instance. If the constructed instance is solvable, we claim that each expression $(\Box + \Box + \Box)$ must have equal value. Denote the value of the $i$th $(\Box + \Box + \Box)$ by $s_i$. Since $\sum_i s_i = S$, the arithmetic mean-geometric mean inequality yields $\prod_{i = 1}^{n / 3} s_i \le \left( \frac{S}{n/3} \right)^{n/3}$, with equality occurring if and only if $s_i = \frac{S}{n/3}$ for all $i$. This completes the proof.

\subsection{\texorpdfstring{$\{+, - \}$}{\{+, −\}}, \texorpdfstring{$\{-\}$}{\{−\}} has a Pseudopolynomial Algorithm for All Variants}
\newcommand{\DPStd}{\mathbf{DPSums}} 

The following pseudopolynomial algorithms demonstrate that the the proofs of weak NP-hardness for all variants of \SG{\mathbb{N}}{\{+, -\}}{\variant} and \SG{\mathbb{N}}{\{-\}}{\variant} are tight.

In all variants, all the possible values which are formable from an AEC instance with integers $\{a_1,\ldots, a_n\}$ are of the form 
$$\sum_i c_i a_i,\ c_i \in \{\pm 1\}.$$
and solvability is determined by whether or not such an expression can equal $t$, the target number.  For example, in \SGSTD{\mathbb{N}}{\{+, -\}}, we can construct all expressions of the above form where at least one $c_i$ is positive (the first number in the expression), whereas in \SGEL{\mathbb{N}}{\{-\}} we must have $c_1 = 1$ and $c_2 = -1$.

The algorithm $\DPStd$ effectively evaluates all expressions of the above form.  Table entry $DP[i, v]$ is true exactly when there is a $c_j\in \{\pm 1\}$ assignment such that the $\sum_{j\le i}c_ja_j = v$.  The innermost for loop iterates across all integers possibly already formed and adds or subtracts $a_i$ from them to set new table entries.  We use $\DPStd$ or slight modifications of it to give pseudopolynomial algorithms for all variants of AEC with $\{+,-\}$ and $\{-\}$.

\begin{algorithm}
\begin{algorithmic}[1]
\Procedure{\textbf{DPSums}}{$\{a_1,\ldots,a_n\}$}
\State $DP[i,v] \gets \text{ new table }$
\State $DP[1,\pm a_1] \gets True$
\For{$i \in [2, n]$}
\For {$v \in [-\sum_{j=0}^{i-1} a_j, \sum_{j=0}^{i-1} a_j]$}
\If{$DP[i-1,v]$} $DP[i,v\pm a_i]\gets True$
\EndIf 
\EndFor
\EndFor
\State \textbf{return} $DP$
\EndProcedure
\end{algorithmic}
\caption{DP Subroutine}
\end{algorithm}

\subsubsection{\texorpdfstring{$\{+, -\}$}{\{+, −\}}, \texorpdfstring{$\{-\}$}{\{−\}} has a Pseudopolynomial Algorithm for Standard and Enforced Leaves}

For \SGSTD{\mathbb{N}}{\{+, -\}}, we can construct any expression of the above form except when all $c_i = -1$, so our algorithm first checks if $t = -\sum_i a_i$ and returns false if so, otherwise returning $\DPStd(\{a_i\})[n,t]$.

For \SGEL{\mathbb{N}}{\{+, -\}}, the only requirement one has is that $c_1$ must equal $1$, so we modify line 3 of $\DPStd$ to only set $DP[1,a_1]$ to true, then return $\DPStd(\{a_i\})[n,t]$.

The \SG{\mathbb{N}}{\{-\}}{\variant} case is slightly more involved. Any expression with purely $-$ operations must have at least one $c_i$ be positive and one $c_i$ be negative. We claim that this is the only restriction and that any other expression is formable. That is, given enforced leaves 
$$a_1 \enspace a_2 \enspace a_3 \enspace \cdots \enspace a_n$$
one can form the expressions 
$$a_1 - a_2 + \sum_{i>2}c_i a_i$$
with $c_i$ arbitrary. To see this, we use the same sign-flipping procedure from Section~\ref{leaves:minus} and note that given a chain of positive coefficients $c_i$ following a negative coefficient we can rewrite the expression using only negative signs via converting: 
$$-a_i+ a_{i+1} + \cdots + a_{i+k} \mapsto -(a_i - a_{i+1} - \cdots - a_{i+k})$$ 
thus, any $\pm 1$-coefficient expression is formable in $\{-\}$-ops given that there is always a $c_i$ which is negative preceding any positive $c_i$ (with the exception of $c_1 = 1$). For Enforced Leaves, this translates to the expression written above, so solvability of an AEC instance is given by $\DPStd(\{a_i\}_{i>2})[n-2,t+a_2-a_1]$. 

For the Standard variant, \SGSTD{\mathbb{N}}{\{-\}}, we can simply loop through all choices of $a, a'\in A$ and return true if any $\DPStd(\{a_i\}-\{a,a'\})[n-2, t+a' - a]$ is true. This will check for all expressions with at least one positive $c_i$ and one negative $c_i$. 

\subsubsection{Enforced Operations \texorpdfstring{$\{+, -\}$}{\{+, −\}}, \texorpdfstring{$\{-\}$}{\{−\}} has a Pseudopolynomial Algorithm}

Trivially, note that \SGEO{\mathbb{N}}{\{-\}} reduces to \SGEO{\mathbb{N}}{\{+, -\}} and thus we need only provide a single algorithm. Any expression $E$ in operations $\{+,-\}$ expands into an expression with $k$ additions and $n-k$ subtractions. Thus, given an instance $I = \{a_i,E\}$, we can compute the number $k$ of additions and then use a similar DP algorithm which keeps count of the number of additions used to form a given subexpression. The algorithm is given by \textbf{DPCount}. 

\begin{algorithm}
\caption{DP Algorithm for $\{+,-\}$-AEC-EO}
\begin{algorithmic}
\Procedure{\textbf{DPCount}}{$\{a_1,\ldots,a_n\},k$}
\State $DP[i,c,v] \gets \text{ new table }$
\State $DP[1,k-1,a_1] = True$
\State $DP[1,k,-a_1] = True$
\For{$i \in [2, n]$}
\For{$c \in [0, k]$}
\For {$v \in [-\sum_{j=0}^{i-1} a_j, \sum_{j=0}^{i-1} a_j]$}
\If{$DP[i-1,c,v]$}
$DP[i,c-1,v+a_i]\gets True,$ $DP[i,c,v-a_i]\gets True$ 
\EndIf 
\EndFor
\EndFor
\EndFor
\State \textbf{return} $DP[n,0,0]$
\EndProcedure
\end{algorithmic}
\end{algorithm}

\section*{Acknowledgments}
{This work was initiated during open problem solving in the MIT class on
Algorithmic Lower Bounds: Fun with Hardness Proofs (6.892)
taught by Erik Demaine in Spring 2019.
We thank the other participants of that class ---
in particular, Josh Gruenstein, Mirai Ikebuchi, and Vilhelm Andersen Woltz ---
for related discussions and providing an inspiring atmosphere.}

\bibliography{references}

\begin{thebibliography}{NBCK10}

\bibitem[AB09]{AroraBarak}
Sanjeev Arora and Boaz Barak.
\newblock {\em Computational Complexity: A Modern Approach}.
\newblock Cambridge University Press, USA, 2009.

\bibitem[AM09]{AggarwalMaurer}
Divesh Aggarwal and Ueli Maurer.
\newblock Breaking {RSA} generically is equivalent to factoring.
\newblock In Antoine Joux, editor, {\em Advances in Cryptology --- EUROCRYPT
  2009}, pages 36--53, Berlin, Heidelberg, 2009. Springer Berlin Heidelberg.

\bibitem[BU08]{BuchfuhrerBooleanFormulaMinimization}
David Buchfuhrer and Christopher Umans.
\newblock The complexity of boolean formula minimization.
\newblock In Luca Aceto, Ivan Damg{\aa}rd, Leslie~Ann Goldberg, Magn{\'u}s~M.
  Halld{\'o}rsson, Anna Ing{\'o}lfsd{\'o}ttir, and Igor Walukiewicz, editors,
  {\em Automata, Languages and Programming}, pages 24--35, Berlin, Heidelberg,
  2008. Springer Berlin Heidelberg.

\bibitem[Dij61]{dijkstra1961algol}
E.~W. Dijkstra.
\newblock {ALGOL-60} translation.
\newblock Technical Report MR 34/61, Rekenafdeling, Stichting Mathematisch
  Centrum, 1961.

\bibitem[DLS81]{minadditionchains}
Peter Downey, Benton Leong, and Ravi Sethi.
\newblock Computing sequences with addition chains.
\newblock {\em SIAM Journal on Computing}, 10(3):638--646, 1981.

\bibitem[GJ02]{GareyJohnson}
Michael~R. Garey and David~S. Johnson.
\newblock {\em Computers and Intractability}.
\newblock W. H. Freeman and Company, New York, 2002.

\bibitem[Gor98]{FastExponentiation}
Daniel~M. Gordon.
\newblock A survey of fast exponentiation methods.
\newblock {\em Journal of Algorithms}, 27:129--146, 1998.

\bibitem[HS11]{HemaspaandraBooleanFormulas}
Edith Hemaspaandra and Henning Schnoor.
\newblock Minimization for generalized boolean formulas.
\newblock arXiv:1104.2312, 2011.

\bibitem[KyC00]{Kabanets99circuitminimization}
Valentine Kabanets and Jin yi~Cai.
\newblock Circuit minimization problem.
\newblock In {\em Proceedings of the 32nd Annual ACM Symposium on Theory of
  Computing}, pages 73--79, Portland, OR, 2000.

\bibitem[Mau05]{Maurer}
Ueli Maurer.
\newblock Abstract models of computation in cryptography.
\newblock In Nigel~P. Smart, editor, {\em Cryptography and Coding}, pages
  1--12, Berlin, Heidelberg, 2005. Springer Berlin Heidelberg.

\bibitem[NBCK10]{ProductPartition}
C.~T. Ng, M.~S. Barketau, T.~C.~E. Cheng, and Mikhail~Y. Kovalyov.
\newblock ``{P}roduct {P}artition'' and related problems of scheduling and
  systems reliability: Computational complexity and approximation.
\newblock {\em European Journal of Operational Research}, 207(2):601--604,
  2010.

\bibitem[Sch37]{scholzchains}
Arnold Scholz.
\newblock {Aufgaben und L{\"{o}}sungen 253}.
\newblock {\em Jahresbericht der Deutschen Mathematiker-Vereinigung},
  47:41--42, 1937.

\bibitem[Sho97]{Shoup}
Victor Shoup.
\newblock Lower bounds for discrete logarithms and related problems.
\newblock In Walter Fumy, editor, {\em Advances in Cryptology --- EUROCRYPT
  '97}, pages 256--266, Berlin, Heidelberg, 1997. Springer Berlin Heidelberg.

\bibitem[vzG88]{Gathen-survey}
Joachim von~zur Gathen.
\newblock Algebraic complexity theory.
\newblock In {\em Annual Review of Computer Science}, volume~3, pages 317--347.
  Annual Reviews Inc., 1988.

\bibitem[Wik]{24game-wiki}
Wikipedia.
\newblock 24 game.
\newblock \url{https://en.wikipedia.org/wiki/24_Game}.

\bibitem[You84]{linearfractional}
N.~J. Young.
\newblock Linear fractional transformations in rings and modules.
\newblock {\em Linear Algebra and its Applications}, 56:251--290, 1984.

\end{thebibliography}

\appendix

\section{Related Problems}
\label{sec:related}

To show the NP-hardness of the variants of Arithmetic Expression Construction, we reduce from the following problems:

\defineproblemrefcomment{\Partition}
{A multiset of positive integers $A = {a_1, a_2, \dots, a_n}$.}
{Can $A$ be partitioned into two subsets with equal sum?}
{\cite{GareyJohnson}, problem SP12.}
{Weakly NP-hard.}

\defineproblemrefcomment{\PartitionEqual}
{A multiset of positive integers $A = {a_1, a_2, \dots, a_n}$.}
{Can $A$ be partitioned into two subsets with equal size $\frac{n}{2}$ and equal sum?}
{\cite{GareyJohnson}, problem SP12.}
{Weakly NP-hard.}

\defineproblemrefcomment{\ProductPartition}
{A multiset of positive integers $A = {a_1, a_2, \dots, a_n}$.}
{Can $A$ be partitioned into two subsets with equal product?}
{\cite{ProductPartition}.}
{Strongly NP-hard.}

\defineproblemcomment{\ProductPartitionEqual}
{A multiset of positive integers $A = {a_1, a_2, \dots, a_n}$.}
{Can $A$ be partitioned into two subsets with equal size $\frac{n}{2}$ and equal product?}
{Strongly NP-hard.  See Theorem~\ref{thm:productpartitionequal}.}

\defineproblemcomment{\SquareProductPartition}
{A multiset of square numbers $A = {a_1, a_2, \dots, a_n}$.}
{Can $A$ be partitioned into two subsets with equal product?}
{Strongly NP-hard.  See Theorem~\ref{thm:squareproductpartitionequal}.}

\defineproblemcomment{\SquareProductPartitionEqual}
{A multiset of square numbers $A = {a_1, a_2, \dots, a_n}$.}
{Can $A$ be partitioned into two subsets with equal size $\frac{n}{2}$ and equal product?}
{Strongly NP-hard.  See Theorem~\ref{thm:squareproductpartitionequal}.}

\defineproblemrefcomment{\SetProductPartitionBound{K}}
{A set (without repetition) of positive integers $A = {a_1, a_2, \dots, a_n}$ where $a_i > K$ and all prime factors of $a_i$ are also greater than $K$. $K$ is fixed and the prime factors are not specified in the instance.}
{Can $A$ be partitioned into two subsets with equal product?}
{\cite{ProductPartition}.}
{Strongly NP-hard by a modification of the proof for \ProductPartition in \cite{ProductPartition}. 
The reduction constructs a set of positive integers $A$ where all elements are unique, which we modify by choosing primes factors $> K$ when constructing $A$.} 

\defineproblemrefcomment{\TPartition-3}
{A multiset of positive integers $A = {a_1, a_2, \dots, a_n}$, with $n$ a multiple of 3.}
{Can $A$ be partitioned into $n/3$ subsets with equal sum, where all subsets have size 3?}
{\cite{GareyJohnson}, problem SP15.}
{Strongly NP-hard, even when all subsets are required to have size 3 (\TPartitionEqual).}

\begin{theorem}\label{thm:productpartitionequal}
\ProductPartitionEqual is strongly NP-complete.
\end{theorem}
\begin{proof}
We can reduce from \ProductPartition to \ProductPartitionEqual. Given instance of \ProductPartition $\{a_1, \cdots, a_n\}_i$ with $n$ elements, where $n$ is even, we construct an corresponding instance of \ProductPartitionEqual as $\{a_1, \cdots, a_n\} \cup \{1\}*n$, where $\{1\}*n$ denotes $n$ instances of the integer $1$.

Clearly if we have a valid solution to \ProductPartitionEqual, we have a valid solution to the instance of \ProductPartition. Conversely, given a valid solution to \ProductPartition, two subsets $S_1, S_2 \subset{\{a_i\}_i}$ with equal product, the difference between the sizes of $S_1$ and $S_2$ is at most $n-2$. One can then distribute the $1$s as needed to even the out the number of elements of $S_1$ and $S_2$. We can then construct two sets: $S_1 \cup \{1\}*|S_2|,\ S_2 \cup \{1\}*|S_1|$ which form a solution to \ProductPartitionEqual. Strong NP-hardness follows from strong NP-hardness of $\ProductPartitionEqual$.
\end{proof}

\begin{theorem}\label{thm:squareproductpartitionequal}\SquareProductPartition and
\SquareProductPartitionEqual is strongly NP-complete.
\end{theorem}
\begin{proof}
One can reduce from $\ProductPartition$ to $\SquareProductPartition$ by simply taking an instance $I = \{a_i\}_{i\in \alpha}$ and producing the instance $I' = \{a_i^2\}_{i\in\alpha}$. Given a partition of $\alpha = \alpha_1\sqcup \alpha_2$ such that $\prod_{i\in \alpha_1}a_i = \prod_{i\in\alpha_2} a_i$, the same partition of $\alpha$ will produce a valid partition of $I'$ as the squares will remain equal. The converse also holds by taking noting that $\prod_{i\in \alpha'}a_i = \sqrt{\prod_{i\in\alpha'} a_i^2}$. The same construction above, with the added requirement that $|\alpha_1| = |\alpha_2|$, will reduce from $\ProductPartitionEqual$ to $\SquareProductPartitionEqual$. Strong NP-hardness of both holds by noting that squaring integers scales their bitsize by a factor of $2$. 
\end{proof}

\end{document}